\documentclass{elsarticle}

\newcommand{\myqed}{\hfill $\blacksquare$}
\usepackage{graphicx}
\usepackage{hyperref}
\usepackage{tabularx,xparse}
\usepackage{amsmath,ntheorem,amssymb}
\usepackage{latexsym}
\usepackage{graphicx} 
\usepackage{multirow}
\usepackage{microtype}
\usepackage{rotating}
\usepackage{xurl}
\usepackage{color}
\usepackage{comment}
\usepackage{paralist}
\usepackage{booktabs}
\usepackage{bm}
\usepackage[]{units}
\usepackage{tikz}
\usepackage{natbib}
\usepackage{tcolorbox}
\usepackage{complexity,multicol}

\newcolumntype{L}{>{$}c<{$}}
\usepackage{colortbl,smartdiagram}
\usesmartdiagramlibrary{additions}
\definecolor{C1}{RGB}{215,215,215}

\newcommand{\rank}{\operatorname{rank}}
\usepackage{mathtools}
\DeclarePairedDelimiter\ceil{\lceil}{\rceil}
\DeclarePairedDelimiter\floor{\lfloor}{\rfloor}
\usepackage{tikz}
\usetikzlibrary{arrows,decorations.markings,patterns,positioning,shapes,decorations.pathreplacing,calc}
\tikzset{draw half paths/.style 2 args={%
  decoration={show path construction,
    lineto code={
      \draw [#1] (\tikzinputsegmentfirst) -- 
         ($(\tikzinputsegmentfirst)!0.5!(\tikzinputsegmentlast)$);
      \draw [#2] ($(\tikzinputsegmentfirst)!0.5!(\tikzinputsegmentlast)$)
        -- (\tikzinputsegmentlast);
    }
  }, decorate
}}
\tikzstyle{vertex} = [circle, draw=black, fill=black, scale= 0.5]
\usetikzlibrary{decorations.pathmorphing}
\tikzset{snake it/.style={decorate, decoration=snake}}

\usetikzlibrary{arrows.meta}
\tikzset{%
  >={Latex[width=2mm,length=2mm]},
            base/.style = {rectangle, rounded corners, draw=black,
                           minimum width=4cm, minimum height=0.6cm,
                           text centered},
  activityStarts/.style = {base, fill=blue!30},
       startstop/.style = {base, fill=red!30},
    activityRuns/.style = {base, fill=green!30},
         process/.style = {base, minimum width=2.5cm, fill=orange!15},
}

\newtheorem{theorem}{Theorem}[section]

\newtheorem{definition}[theorem]{Definition}
\newtheorem{lemma}[theorem]{Lemma}

\newtheorem{corollary}[theorem]{Corollary}
\newtheorem{observation}[theorem]{Observation}

\newproof{proof}{Proof}
\newtheorem{myclaim}[theorem]{Claim}

\usepackage[textsize=tiny,textwidth=2cm,shadow]{todonotes}

\newcommand{\skcom}[1]{\todo[color=orange!50!white]{{S}: #1}}
\newcommand{\dmrev}[1]{\textcolor{black}{#1}}

\newcommand\fauxsc[1]{\fauxschelper#1 \relax\relax}
\def\fauxschelper#1 #2\relax{%
  \fauxschelphelp#1\relax\relax%
  \if\relax#2\relax\else\ \fauxschelper#2\relax\fi%
}
\def\Hscale{.85}\def\Vscale{.72}\def\Cscale{1.10}
\def\fauxschelphelp#1#2\relax{%
  \ifnum`#1>``\ifnum`#1<`\{\scalebox{\Hscale}[\Vscale]{\uppercase{#1}}\else%
    \scalebox{\Cscale}[1]{#1}\fi\else\scalebox{\Cscale}[1]{#1}\fi%
  \ifx\relax#2\relax\else\fauxschelphelp#2\relax\fi}

\journal{Discrete Applied Mathematics}
\begin{document}
\begin{frontmatter}
\title{{\bf On weakly and strongly popular rankings}\tnoteref{t1}
}
\tnotetext[t1]{A 2-page extended abstract of this paper appeared at AAMAS 2021.}




\author[1]{Sonja Kraiczy\corref{cor1}}
\ead{sonja.kraiczy@cs.ox.ac.uk}
\author[2,3]{\'{A}gnes Cseh}
\ead{cseh.agnes@krtk.hu}
\author[3]{David Manlove}
\ead{david.manlove@glasgow.ac.uk}
\cortext[cor1]{Corresponding author}
\address[1]{Department of Computer Science, University of Oxford, UK}
\address[2]{Institute of Economics, Centre for Economic and Regional Studies, Hungary}
\address[3]{Department of Mathematics, University of Bayreuth, Germany}
\address[4]{School of Computing Science, University of Glasgow, UK}



\begin{abstract}
    Van Zuylen et al.~\cite{VSW14} introduced the notion of a popular ranking in a voting context, where each voter submits a strict ranking of all candidates. A popular ranking $\pi$ of the candidates is at least as good as 
    any other ranking $\sigma$ in the following sense: if we compare $\pi$ to $\sigma$, at least half of all voters will always weakly prefer~$\pi$. Whether a voter prefers one ranking to another is calculated based on the Kendall distance.

    A more traditional definition of popularity---as applied to popular matchings, a well-established topic in computational social choice---is stricter, because it requires at least half of the voters \emph{who are not indifferent between $\pi$ and $\sigma$} to prefer~$\pi$. In this paper, we derive structural and algorithmic results in both settings, also improving upon the results in~\cite{VSW14}. We also point out connections to the famous open problem of finding a Kemeny consensus with \dmrev{three} voters.
\end{abstract}

\begin{keyword}
majority rule \sep
Kemeny consensus \sep
complexity \sep
preference aggregation
\sep popular matching
\end{keyword}
\end{frontmatter}
\section{Introduction}

A fundamental question in preference aggregation is the following: given a number of voters who rank candidates \dmrev{from most-preferred to least-preferred}, can we construct a ranking that expresses the preferences of the entire set of voters as a whole? A common way of evaluating how close the constructed ranking is to a voter's preferences is the Kendall distance\dmrev{~\cite{Ken38}}, which measures the pairwise disagreements between two rankings. Among others, a well-known rank aggregation method is the Kemeny ranking method~\cite{Kem59}, in which the winning ranking minimises the sum of its Kendall distances to the voters' rankings. 

For the preference aggregation problem, van Zuylen et al.~\cite{VSW14} introduce a new rank aggregation method called popular ranking, which is also based on the Kendall distance. Each voter can compare two given rankings $\pi$ and $\sigma$, and prefers the one that is closer to her submitted ranking in terms of the Kendall distance. Van Zuylen et al.\ define $\pi$ to be a winning ranking for a given \dmrev{set of voters' rankings} 
if for any ranking $\sigma$, at least half of the voters prefer $\pi$ to $\sigma$ or are indifferent between them. This implies that there is no ranking $\sigma$ such that switching to $\sigma$ from $\pi$ would benefit a majority of all voters. 

According to the definition of popularity in~\cite{VSW14}, even in a situation where exactly half of the voters are indifferent between rankings $\pi$ and $\sigma$\dmrev{---we call these abstaining voters---}, whilst the other half of the voters prefer $\sigma$ to $\pi$, the ranking $\sigma$ is not more popular than~$\pi$. This example demonstrates how \dmrev{challenging} 
it is for the dissatisfied voters to \dmrev{propose} 
a ranking that overrules $\pi$---the definition requires them to find a profiting set of voters who build an \emph{absolute majority}, that is, a majority of \emph{all} voters for this endeavour.

A straightforward option would be to require only a \emph{simple majority}, this is, a majority of the \emph{non-abstaining} voters, to profit from switching to $\sigma$ from~$\pi$. Excluding the abstaining voters in a pairwise majority voting rule is common practice~\cite{DE10}. It is also analogous to the classical popularity notion in the matching literature~\dmrev{\cite{AIKM07}}. In this paper, we propose an alternative definition of a popular ranking. We define $\pi$ to be a  \emph{strongly popular ranking} if for every ranking $\sigma$, at least half of the non-abstaining voters prefer $\pi$ to~$\sigma$. This means that switching from $\pi$ to $\sigma$ would harm at least as many voters as it would benefit.  \dmrev{The weaker notion of a popular ranking as defined by van Zuylen et al.~\cite{VSW14} is then defined here as a \emph{weakly popular ranking}.}

\subsection{\dmrev{Our contribution}}
\label{sec:cont}
We study both the weaker notion of popularity from~\cite{VSW14} and the stronger notion of popularity analogous to the one in the matching literature, which excludes abstaining voters. Our most important results are as follows.
\begin{enumerate}
    \item \dmrev{We give a sufficient condition for the existence of a \dmrev{weakly popular ranking (Theorem~\ref{tightc}).}}
    \item We give a sufficient condition for the two popularity notions to be equivalent for a given ranking~$\pi$ \dmrev{(Lemma~\ref{condeq}).} \dmrev{This condition also implies that at least three abstaining voters between two rankings are needed for the two popularity notions to differ, which allows us to deduce the following \dmrev{result}. 
    }
    \item  For at most \dmrev{five} voters, the two \dmrev{popularity} notions are equivalent \dmrev{(Theorem~\ref{atmost5}).} \dmrev{We provide an example showing that with six voters this equivalence does not hold anymore (Theorem~\ref{example}).}
    \item In the case of \dmrev{two} or \dmrev{three} voters, one can find a popular ranking of either kind and verify the weak or strong popularity of a given ranking in polynomial time \dmrev{(Lemmas~\ref{cl:twovoters} and~\ref{3voters}, Theorem~\ref{th:smallk}).}
    \item The problem of verifying the weak or strong popularity of a given ranking is polynomial-time solvable for \dmrev{four} voters if and only if it is polynomial-time solvable for \dmrev{five} voters \dmrev{(Theorem~\ref{45})}.
    \item \dmrev{Finally, we establish a connection to a central open problem in preference aggregation.}\dmrev{ We show that if} finding a ranking that is more popular in either sense than a given ranking in a \dmrev{given set of four (or five) voters' rankings} were polynomial-time solvable, then the famously open Kemeny consensus problem for \dmrev{three} voters would also be polynomial-time solvable \dmrev{(Corollary~\ref{cor:kemeny}). }
    \dmrev{ The path to this result leads through an even stronger observation: If finding a ranking preferred to a given ranking by \emph{all} the \dmrev{three}  
    voters were polynomial-time solvable, then the Kemeny consensus problem for three voters would also be polynomial-time solvable \dmrev{(Theorems~\ref{th:Kc} \dmrev{ and \ref{th:Kr})}}.}
\end{enumerate}

\subsection{Related literature} Aggregating voters' preferences given as rankings of candidates has been challenging researchers for decades. The most common approach to this problem is to search for a ranking that minimises the sum of the distances to the voters’ rankings. If the Kendall distance~\cite{Ken38} is used as the metric on rankings, then this optimality concept corresponds to the Kemeny consensus~\dmrev{\cite{Kem59}}. 
\dmrev{Characteristic properties and computational aspects of the Kemeny consensus have been studied in a number of papers~\cite{BTT89,DK04,HSV05,BFG+09,ADS16}. The problem setting \dmrev{with} a small number of voters received special interest.} Deciding whether a given ranking is a Kemeny consensus is $\coNP$-complete~\cite{FH21}, and calculating a Kemeny consensus is \dmrev{$\NP$}-hard~\cite{BTT89} \dmrev{even if there is a fixed number of voters $k$, where $k=7$~\cite{BBG+19}, or where $k$ is even and $k\geq 4$~\cite{DKN+01}}. The complexity of the problem for \dmrev{three} and \dmrev{five} voters is pointed out as an interesting open problem in~\cite{BBF+09,BBG+19}. 
\dmrev{ Milosz et al.~\cite{MHP18} focus on the case of three voters, and establish a link with the 3-Hitting Set problem~\cite{Kar72} by  considering 3-cycles in the majority graph.
}

Majority voting rules offer another natural way of aggregating voters' preferences. The earliest reference for this might be from Condorcet~\cite{Con85}, who uses pairwise comparisons to calculate the winning candidate, establishing his famous paradox on the smallest \dmrev{set of voters' rankings} not admitting any majority winner.

In some settings, handling abstaining voters plays a crucial role. The absolute and simple majority voting rules have both been extensively discussed in the setting where the goal is to choose the winning candidate~\cite{BL14a,BL14b}. Vermeule~\cite{Ver07} focuses on strategic minorities and demonstrates the effect of the simple majority rule compared to the absolute majority rule based on data from decisions made by the United States Congress. By undertaking a probabilistic analysis, Dougherty and Edward~\cite{DE10} discuss the differences between the two rules. Felsenthal and Machover~\cite{MF97} generalise simple voting games to ternary voting games by adding the possibility to abstain. 

The concept of majority voting readily translates to other scenarios, where voters submit preference lists. One such field is the area of matchings under preferences, where popular matchings~\citetext{\citealp{Gar75},~\citealp{AIKM07},~\citealp[Chapter~7]{Man13},~\citealp{Cse17}} serve as a voting-based alternative concept to the well-known notion of stable matchings~\cite{GS62} in two-sided markets. In short, a popular matching $M$ is a simple majority winner among all matchings, because it guarantees that no matter what alternative matching is offered on the market, \dmrev{at least half} of the non-abstaining \dmrev{voters} 
will opt for~$M$. 

Besides two-sided matchings, majority voting has also been defined for the house allocation problem~\cite{AIKM07,SM10}, the roommates problem~\cite{FKP+19,GMS+19}, spanning trees~\cite{Dar13}, permutations~\cite{VSW14}, the ordinal group activity selection problem~\cite{Dar18}, and very recently, for branchings~\cite{KKM+20}. The notion of popularity is aligned with simple majority in all these papers, with one exception, namely~\cite{VSW14}, which defines popularity based on the absolute majority rule.

A part of this work revisits the paper from van Zuylen et al.~\cite{VSW14}. They show that a popular ranking---according to their definition of popularity---need not necessarily exist. More precisely, they show that the acyclicity of a structure known as the majority graph is a necessary, but not sufficient condition for the existence of a popular ranking. They also prove that if the majority graph is acyclic, then \dmrev{we} can efficiently \dmrev{compute a ranking (corresponding to a toplogical sort of the majority graph)}, which may or may not be popular, but for which the voters have to solve an $\NP$-hard problem to compute a ranking that a majority of them prefer. 

\subsection{Structure of the paper} In Section~\ref{prelims} we introduce the necessary definitions and notations used in the following sections. Section~\ref{strandweak} deals with the relationship between the two different popularity notions. In Section~\ref{smallvoters} we study the complexity of the problems of deciding whether a given ranking is weakly or strongly popular with a small number of voters. We demonstrate the strong connection to the Kemeny consensus problem \dmrev{in Section~\ref{sec:kemeny}. Finally, we} lay out some problems that still remain to be answered in Section~\ref{summary}. \dmrev{Throughout this paper, we use the example instance depicted in Figure~\ref{fi:ex} to illustrate concepts that we will define.}

\section{Preliminaries}
\label{prelims}

We start this section with the formal definitions of various standard notions in voting theory in Section~\ref{se:rankings}. Then, in Section~\ref{se:preferences}, we introduce weakly and strongly popular rankings and the decision problems we will later analyse.

\subsection{Rankings \dmrev{and Kendall distance}}
\label{se:rankings}
We are given a set $C=\{1,\ldots,m\}$ of candidates and a set \dmrev{$V=\{v_1,\ldots, v_n\}$} of $n$ voters. A (preference) \emph{ranking} $\pi$ is a \dmrev{total} 
order over~$C$. When exhibiting a specific ranking, we will often enclose parts of it in square brackets, e.g.\ we may write $[1,2],[3,4]$ instead of $1,2,3,4$. These brackets can be ignored and are simply used for better readability in \dmrev{sets of rankings} with specific structural properties.
The \emph{rank} of candidate $a$ in ranking $\pi$ is the position (counting from~1) it appears at in $\pi$, and it is denoted by $\rank_{\pi}[a]$.  A 
\textit{profile}~$P=(\pi_{v_1},\ldots,\pi_{v_n})$ \dmrev{over $C$} is a list of rankings
, where $\pi_{v_i}$ is the ranking associated with voter $v_i\in V$. 
 An example is depicted in Figure~\ref{fi:ex}. 
 We say that voter $v$ \emph{prefers} candidate $a$ to candidate $b$ if $\rank_{\pi_v}[a] < \rank_{\pi_v}[b]$. 
In Figure~\ref{fi:ex}, voter $v_1$ prefers candidate $1$ to candidate $2$, and $\rank_{\pi_{v_{1}}}[5]=6$. 
For candidates $a$ and $b$, $\{a,b\}$ denotes the unordered pair of them, while $(a,b)$ denotes an ordered pair.

\begin{figure}[th]
\begin{align*} 
\pi_{v_{1}} & = [1,2,3],[6,4,5],[8,9,7] \\ 
\pi_{v_{2}} & =[2,3,1],[4,5,6],[9,7,8] \\
\pi_{v_{3}} & =[3,1,2],[5,6,4],[7,8,9] \\
\pi_{v_{4}} & =[1,2,3],[4,5,6],[7,8,9] \\
\pi_{v_{5}} & = [1,2,3],[5,4,6],[9,7,8] \\
\pi_{v_{6}} & = [1,2,3],[5,6,4],[7,9,8] 
\end{align*}
\caption{A 
profile $P=(\pi_{v_1}\ldots,\pi_{v_6})$ over $\{1,2,\ldots,9\}$.}
\label{fi:ex}
\end{figure}

We say that voter $v_1$ (or ranking $\pi_{v_1}$) \emph{agrees} with voter $v_2$ (or with ranking $\pi_{v_2}$) in the order of two distinct candidates $a$ and $b$ if $v_1$ and $v_2$ either both prefer $a$ to $b$ or they both prefer $b$ to~$a$. 
Otherwise they \emph{disagree} in the order of $a$ and~$b$. The similarity between two rankings can be measured by various metrics defined on permutations. Possibly the most common metric, the Kendall distance~\cite{Ken38}, is defined below.  

\begin{definition}
\label{def:kd}
The \emph{Kendall distance} $K(\pi,\sigma)$ between two 
rankings $\pi$ and $\sigma$ is defined as the number of pairwise disagreements between $\pi$ and $\sigma$, or, formally as
\begin{eqnarray*}
K(\pi,\sigma) = \dmrev{|\{(a,b)\dmrev{\in C\times C} : \rank_{\pi}[a]>\rank_{\pi}[b] \mbox{ and } \rank_{\sigma}[a]<\rank_{\sigma}[b]\}|.}\end{eqnarray*}
\end{definition}
\dmrev{Alongside Definition \ref{def:kd},} the Kendall distance \dmrev{has an interpretation, which we will also use, in terms of the bubble sort algorithm~\cite{Fri56}. Given a sequence $\langle \sigma_1,\sigma_2,\dots,\sigma_m\rangle$, bubble sort proceeds inductively by considering the $i$th element $\sigma_i$ ($i=1,2,\dots,m$), and assuming that $\langle\sigma_1,\sigma_2,\dots \sigma_{i-1}\rangle$ are already in the correct order, $\sigma_i$ is swapped with its predecessor in $\langle\sigma_1,\sigma_2,\dots,\sigma_{i-1},\sigma_i\rangle$ as long as $\sigma_i$ is larger than its predecessor in this subsequence.  The Kendall distance is also called the \emph{bubble sort distance}} \dmrev{because it corresponds to the number of swaps that bubble sort} executes when converting ranking $\pi$ to ranking~$\sigma$. To be more precise, let us first define a total order on $1,\ldots, n$ such that, under this order, ranking $\sigma$ is sorted in increasing order. 
We define the \textit{bubble swap path} from a ranking $\pi$ to $\sigma$ as the sequence $\pi_{0}:=\pi, \pi_{1},\ldots, \pi_{k}:=\sigma$ of intermediate rankings obtained when sorting $\pi$ using the bubble sort algorithm. \dmrev{Note that $K(\pi,\sigma)=k$.} We call the change $\pi_{i}\rightarrow \pi_{i+1}$ a \textit{swap}. Alternatively we denote the swap  by the consecutive candidates $a$ and $b$ it interchanges: $(b,a)\rightarrow (a,b)$. We say that a swap $(b,a)\rightarrow(a,b)$ is \emph{good} for voter $v$ if $v$ prefers $a$ to $b$, otherwise this swap is \emph{bad} for~$v$. Note that if the swap $\pi_{i}\rightarrow \pi_{i+1}$ is good for $v$, then $K(\pi_{i+1},\pi_{v})=K(\pi_{i},\pi_{v})-1$ and, analogously, if the swap is bad for $v$, then $K(\pi_{i+1},\pi_{v})=K(\pi_{i},\pi_{v})+1$.

\dmrev{For example, with respect to the profile shown in Figure \ref{fi:ex}, $K(\pi_{v_1},\pi_{v_2})=6$, since it takes two swaps to insert each of $1$, $6$ and $8$ into the correct order relative to $[2,3,1]$, $[4,5,6]$ and $[9,7,8]$, respectively, in the bubble swap path from $\pi_{v_1}$ to~$\pi_{v_2}$.}

Let $V(a,b) \subseteq V$ be the set of voters who prefer candidate $a$ to $b$, i.e.\ $V(a,b)=\{v\in V: \rank_{\pi_{v}}(a) < \rank_{\pi_{v}}(b)\}$. The \textit{majority graph} belonging to a profile is defined as the directed graph which has as vertices the candidates and an arc from candidate $a$ to candidate $b$ if a majority of the voters prefer $a$ to $b$, i.e.\ $|V(a,b)|>|V(b,a)|$. As mentioned in the introduction, Condorcet observed that the majority graph may contain a directed cycle. This has come to be known as the Condorcet paradox~\cite{Con85}. A \textit{tournament} is a majority graph that is complete, or, in other words, for every $a$ and $b$ either $|V(a,b)|>|V(b,a)|$ or $|V(a,b)|<|V(b,a)|$ holds. The majority graph of our profile in Figure~\ref{fi:ex} is depicted in Figure~\ref{fi:majority}. As the edges form no cycle, it is an acyclic majority graph, but since it is not a complete graph, it is not a tournament.

\begin{figure}[htb]
    \centering    
    \begin{tikzpicture}[scale=.5]
  \tikzstyle{every ellipse node}=[draw,inner xsep=3.5em,inner ysep=1em,fill=black!15!white,draw=black!15!white]
  \foreach \x / \y in {3/2,8/1,5.5/1a}
    \draw  (2,\x) node[ellipse] (ellipse\y){};
  \tikzstyle{every circle node}=[draw,minimum size=1.6em,inner sep=0pt,fill=white]
  \draw (0,3)		node[circle](7){$7$} ++(2,0) node[circle](8){$8$} ++(2,0) node[circle](9){$9$};
  \draw (0,5.5)		node[circle](4){$4$} ++(2,0) node[circle](5){$5$} ++(2,0) node[circle](6){$6$};
  \draw (0,8)		node[circle](1){$1$} ++(2,0) node[circle](2){$2$} ++(2,0) node[circle](3){$3$};

   \foreach \i/\j in {2/3,1/2,7/8,5/6} {\draw[-latex] (\i) -- (\j);}

   \draw[-latex] (1.45) to [out=30,in=150] (3.135);
   \draw[-latex] (ellipse1) to [out=185,in=180] (ellipse1a);
   \draw[-latex] (ellipse1) to [out=180,in=180] (ellipse2);
   \draw[-latex] (ellipse1a) to [out=0,in=0] (ellipse2);
\end{tikzpicture}
    \caption{The majority graph of the profile from Figure~\ref{fi:ex}. Each \dmrev{grey} set of \dmrev{three} candidates denotes a specific subgraph \dmrev{that satisfies the same property as a \emph{component} in a tournament}~\cite{Las97}. Arcs between these components symbolise the complete set of 9 arcs, between any two vertices from different components.}
    \label{fi:majority}
\end{figure}

Ranking $\pi$ is a \emph{topologically sorted ranking} of profile $P$ if 
$\rank_{\pi}[a] < \rank_{\pi}[b]$  holds for each pair of candidates $a$ and $b$ with $|V(a,b)| > |V(b,a)|$. Topologically sorted rankings correspond to the graph-theoretical topological sort of the vertices in the majority graph, and thus only exist if the majority graph is acyclic. Acyclic tournaments trivially have a unique topologically sorted ranking. A topologically sorted ranking for the profile in Figure~\ref{fi:ex} with $9$ candidates is $\sigma=[1,2,3],[4,5,6],[7,8,9]$, as can be checked easily.

The \emph{Kemeny rank} of a ranking $\pi$ for a given profile with voters $v_{1},\ldots,v_{n}$ is defined as $\sum_{i=1}^{n} K(\pi,\pi_{v_{i}})$. If ranking $\sigma$ has minimum Kemeny rank over all rankings, then $\sigma$ is a \textit{Kemeny consensus}~\cite{Kem59}. The following well-known observation~\cite{DK04} will be useful in our proofs. 
\begin{observation}\label{fact2}
     Each topologically sorted ranking is a Kemeny consensus. For acyclic majority graphs, the set of topologically sorted rankings coincides with the set of \dmrev{Kemeny consensuses}. 
\end{observation}

\subsection{\dmrev{Popularity and problem definitions}}
\label{se:preferences}

\paragraph{\dmrev{Preferences over rankings}} Voters prefer rankings that are similar to their submitted ranking. More precisely, voter $v$ \emph{prefers} ranking $\sigma$ to ranking $\pi$ if $K(\sigma,\pi_{v}) < K(\pi,\pi_{v})$. Analogously, voter $v$ \textit{abstains} 
between $\pi$ and $\sigma$ if $K(\sigma,\pi_{v})=K(\pi,\pi_{v})$. We will simply call $v$ an \textit{abstaining voter} if $\pi$ and $\sigma$ are clear from the context.

For example, let $\sigma_{1}=[1,2,3],[5,6,4],[9,7,8]$ and $\sigma_{2}=[1,2,3],\allowbreak [4,5,6],[7,8,9]$. Looking back at Figure~\ref{fi:ex}, clearly $v_{4}$ prefers $\sigma_2$ to $\sigma_1$, since $\pi_{v_{4}}=\sigma_{2}$ and $\pi_{v_{4}} \neq \sigma_{1}$, that is, $K(\pi_{v_{4}},\sigma_{2})=0<K(\pi_{v_{4}},\sigma_{1})$. Voter $v_{1}$ in the same profile is an abstaining voter since $K(\pi_{v_{1}},\sigma_{2})=4=K(\pi_{v_{1}},\sigma_{1})$. \dmrev{For more detail regarding these calculations, the reader is referred to \ref{app:preferences}.}

\paragraph{\dmrev{Majority concepts}}
\dmrev{We now formally define the two majority concepts we rely on in this paper. A set $V' \subseteq V$ of voters \dmrev{forms} an \emph{absolute majority} of the voters if $|V'| >n/2$.
For the more intricate majority concept, two rankings \dmrev{$\pi$ and $\sigma$ must be given}. Let $V_{\text{abs}}(\pi,\sigma)$ be the set of voters who abstain in the vote between rankings $\pi$ and $\sigma$, that is, $v \in V_{\text{abs}}(\pi,\sigma)$ if and only if $K(\pi_{v},\pi)=K(\pi_{v},\sigma)$. In Figure~\ref{fi:ex}, $V_{\text{abs}}(\sigma_1,\sigma_2)=\{v_{1},v_{2},v_{3} \}$, \dmrev{where $\sigma_1$ and $\sigma_2$ are as defined in the previous paragraph}. A set $V' \subseteq V$ of voters \dmrev{forms} a \emph{simple majority} of the voters if $|V'| > |V \setminus V_{\text{abs}}(\pi,\sigma) |/2$. A simple majority therefore always depends on the rankings that are \dmrev{being} compared.}

\paragraph{\dmrev{Popularity concepts}}
We now define the two different notions of popularity. The first notion of a weakly popular ranking corresponds to the popular ranking as defined in~\cite{VSW14}. Let $C$ be a set of candidates \dmrev{and let} $P$ \dmrev{be} a profile over $C$ of voters $V$.
     \begin{definition}
     \label{def:abs}
        Ranking $\pi'$ is \emph{more popular than ranking $\pi$ in the absolute sense} if $K(\pi', \pi_{v}) < K(\pi, \pi_{v})$ for an absolute majority of all voters $v\in V$. Ranking $\pi$ is \emph{weakly popular} for $P$ if no ranking $\pi'$ is more popular than $\pi$ in the absolute sense, in other words, if there is no ranking $\pi'$ such that $K(\pi', \pi_{v}) < K(\pi, \pi_{v})$ for an absolute majority of all voters in $v\in V$.
     \end{definition}
If we consider $\sigma_3=[2,1,3],[4,5,6],[7,8,9]$, then in the profile in Figure~\ref{fi:ex}, $\sigma_{2}=[1,2,3],[4,5,6],[7,8,9]$ is more popular than $\sigma_3$ in the absolute sense. Notice that $\sigma_3$ and $\sigma_2$ only differ in their ordering of the pair of candidates $\{1,2\}$. So since five out of six voters prefer candidate $1$ to candidate $2$, they form an absolute majority of all voters who prefer $\sigma_{2}$ to~$\sigma_3$.

This definition requires more than half of the $n$ voters to prefer $\pi'$ to $\pi$ in order to declare $\pi'$ to be more popular than~$\pi$. Abstaining voters make it hard to beat $\pi$ in such a pairwise comparison. However, if $\pi'$ only needs to receive more votes than $\pi$ among the voters not abstaining between these two rankings, then it can beat~$\pi$. This leads to the notion of strong popularity.
     \begin{definition}
        Ranking $\pi'$ is \emph{more popular than ranking $\pi$ in the simple sense} if $K(\pi', \pi_v) < K(\pi, \pi_v)$ for \dmrev{an absolute} majority of the non-abstaining voters $v\in V \setminus V_{\text{\emph{abs}}}(\pi,\pi')$. Ranking $\pi$ is \emph{strongly popular} for $P$ if no ranking $\pi'$ is more popular than $\pi$ in the simple sense, in other words, if there is no ranking $\pi'$ such that $K(\pi', \pi_{v}) < K(\pi, \pi_{v})$ for \dmrev{an absolute} majority of the non-abstaining voters $v\in V\setminus V_{\text{\emph{abs}}}(\pi,\pi')$.
     \end{definition}
It follows directly from the two definitions above that strongly popular rankings are weakly popular as well, but weakly popular rankings are not necessarily strongly popular. In the profile in Figure~\ref{fi:ex},  $\sigma_{1}=[1,2,3],[5,6,4],[9,7,8]$ is more popular than $\sigma_{2}=[1,2,3],[4,5,6],[7,8,9]$ in the simple sense, since $v_{5}$ and $v_{6}$ prefer $\sigma_{1}$ to $\sigma_{2}$, while $v_{1},v_{2}$, and $v_{3}$  abstain. Notice that $\sigma_{1}$ is not more popular than $\sigma_{2}$ in the absolute sense, because two voters do not constitute an absolute majority of all six voters, only \dmrev{an absolute} majority of the non-abstaining three voters. \dmrev{Again, we provide more explanation for these calculations in \ref{app:preferences}.}

\paragraph{\dmrev{Problem definitions}}
We now define \dmrev{two} natural verification problems arising from the notions of weakly and strongly popular rankings\dmrev{.}

\begin{tcolorbox}
\textbf{$k$-\fauxsc{weakly-unpopular-ranking-verify} ($k$-\fauxsc{wurv})}\\
\textbf{Input}: A set $C$, a profile $P$ over $C$ of size $k$ and a ranking $\pi$ over $C$.\\
\textbf{Question}: Does there exist a ranking $\sigma$ that is preferred to $\pi$ by \dmrev{an absolute} majority of all voters?
\end{tcolorbox}
\begin{tcolorbox}
\dmrev{
\textbf{$k$-\fauxsc{strongly-unpopular-ranking-verify} ($k$-\fauxsc{surv})}\\
\textbf{Input}: A set $C$, a profile $P$ over $C$ of size $k$ and a ranking $\pi$ over $C$.\\
\textbf{Question}: Does there exist a ranking $\sigma$ that is preferred to $\pi$ by a simple majority of all voters?
}
\end{tcolorbox}
\dmrev{In other words, $k$-\fauxsc{wurv} (respectively $k$-\fauxsc{surv}) asks whether $\pi$ is not weakly popular (respectively not strongly popular).}
\medskip

\dmrev{We derive the voters' preferences over rankings from their preferences over candidates using the notion of Kendall distance. Providing the voters’ preferences over rankings explicitly \dmrev{as part of the input, as a list of all rankings, would increase the number of ranking list entries from $nm$ to $n\cdot m!$}. Our problems $k$-\fauxsc{wurv} and $k$-\fauxsc{surv} can be solved by iterating through \dmrev{every possible ranking $\sigma$ and comparing $\sigma$} to the given ranking~$\pi$. Therefore, \dmrev{with an input model that involves explicit preferences over rankings}, the problems addressed in the paper would be trivially solvable in polynomial time \dmrev{relative} to the significantly increased input size.}

\section{\dmrev{Relationships involving} weakly and strongly popular rankings}
\label{strandweak}
In this section, we study \dmrev{connections between weakly and strongly popular rankings, and between other concepts involving rankings}. We first place weakly and strongly popular rankings in the context of Kemeny consensuses in Section~\ref{se:revisit}. Then in Section~\ref{se:6voters} we show that for as few as six voters, the two notions of popularity \dmrev{are not} equivalent and that at the heart of this lies Condorcet's paradox. Building upon our \dmrev{six}-voter
example \dmrev{from Figure \ref{fi:ex}}, in Section~\ref{se:condition} we derive a sufficient, but not necessary condition for \dmrev{a} weakly popular ranking to be strongly popular as well. This condition opens a way to prove in Section~\ref{se:5voters} that for up to \dmrev{five} voters the two notions are equivalent.

\subsection{\dmrev{Relationships with other properties}}
\label{se:revisit}

We first revisit two results from~\cite{VSW14}, established for weak popularity, and translate them for the notion of strong popularity \dmrev{in Lemma~\ref{condorcetweak}}.

\begin{lemma}\label{condorcetweak}
If a profile $P$ has a majority graph with a directed cycle, then there does not exist a strongly popular ranking. If $P$ has an acyclic majority graph, then a topologically sorted 
ranking is not necessarily strongly popular.
\end{lemma}
\begin{proof}
These two statements hold for weak popularity~\cite[Lemmas~2 and~3]{VSW14} and hence also for strong popularity, because strongly popular rankings are also weakly popular by definition.
\qed\end{proof}



Each weakly popular ranking must be topologically sorted, as shown by the proof of~\cite[Lemma 2]{VSW14}. In short, if the majority graph has a directed cycle, for an arbitrary ranking $\succ$ of the candidates, there will be two candidates $a,b\in C$ such that $a\succ b$ but \dmrev{an absolute} majority of voters prefers $b$ to $a$. One can show that the ranking $\succ'$ obtained by swapping $a$ and $b$ in $\succ$ is preferred to $\succ$ by every voter $v$ satisfying $b\succ_v a$, which is \dmrev{an absolute} majority of all voters.
This result together with Lemma~\ref{condorcetweak} and Observation~\ref{fact2} lead\dmrev{s} to the following set inclusion relationships \dmrev{involving} weakly popular, strongly popular \dmrev{and} topologically sorted \dmrev{rankings}, and \dmrev{Kemeny consensuses}.

\begin{observation}\label{obs:topsort}
Strongly popular rankings form a subset of weakly popular rankings, weakly popular rankings form a subset of topologically sorted rankings, and finally, topologically sorted rankings form a subset of \dmrev{Kemeny consensuses}. In profiles with an acyclic majority graph, topologically sorted rankings coincide with \dmrev{Kemeny consensuses}.
\end{observation}

Figure~\ref{fi:hierarchy} illustrates these relations. In profiles with a cyclic majority graph, \dmrev{Kemeny consensuses} offer a preference aggregation method by relaxing the definition of topologically sorted rankings. Weakly and strongly popular rankings do exactly the opposite: they restrict the set of topologically sorted rankings in profiles with an acyclic majority graph, in order to serve the welfare of the majority to an even larger degree than topologically sorted rankings do. \dmrev{Weak} and strong popularity are desirable robustness properties of a ranking.  However, the set of weakly popular rankings may be empty even if a topologically sorted ranking exists~\cite{VSW14}.

\dmrev{This is reminiscent \dmrev{of} single-winner elections in which being a Condorcet winner is a strong property. 
\dmrev{However, due to} Condorcet's paradox\dmrev{,} such a winner may not exist. 
\dmrev{The Condorcet Criterion, an axiom for voting rules, thus states that the winner of an election should be a Condorcet winner \emph{if one exists}.}
Most single-winner voting rules, such as Kemeny-Young, Black, Copeland, Dodgson's method, Minimax, Nanson's method, Ranked pairs, Schulze, Smith/IRV, Smith/minimax \dmrev{and} CPO-STV satisfy the Condorcet Criterion\dmrev{---for more details on those methods, we refer the reader to~\cite{BCE+16}}. \dmrev{Analogously, we} envision a ``Popularity criterion" for rank aggregation rules, such that a strongly popular ranking should always be chosen if one exists, and failing that, a weakly popular ranking if it exists.}

\begin{figure}[th]
\centering
  \includegraphics[width=0.8\textwidth]{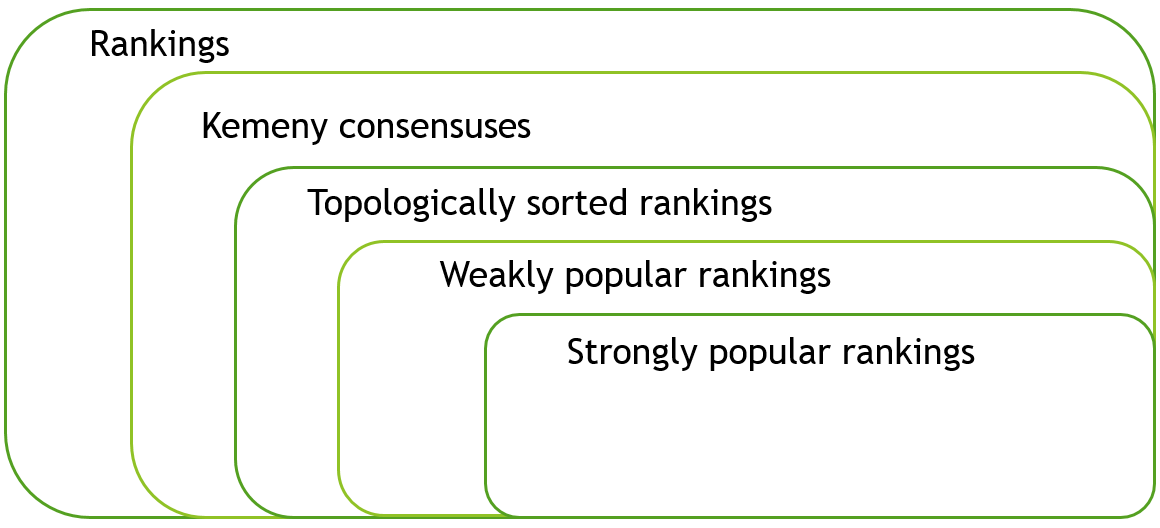}
\caption{The hierarchy of various optimality notions for rankings as a Venn diagram. A \dmrev{Kemeny consensus} always exists, but its subset of topologically sorted rankings might be empty.}
\label{fi:hierarchy}
\end{figure}

\subsection{Difference for \dmrev{six} voters}
\label{se:6voters}

Before presenting our example with $n=6$, we present a useful technical lemma. Let $(C_{1},\ldots, C_{k})$ be an \textit{ordered} partition of $C$ into $k$ sets. We say that a ranking $\pi$ \textit{preserves} $(C_{1},\ldots, C_{k})$ if it holds that rank$_{\pi}[a]<$rank$_{\pi}[b]$ whenever $a \in C_{i}$ and $b \in C_{j}$ for some~$i<j$.

\begin{lemma}\label{pres}
Let $(C_1,\dots,C_k)$ be an ordered partition of $C$ such that, for each $v\in V$, $\pi_v$ preserves $(C_1,\dots,C_k)$.  Then for any ranking $\sigma$, there exists a ranking $\zeta$ such that, for each $v\in V, \zeta$ preserves $(C_1,\dots,C_k)$ and $v$ prefers $\zeta$ to $\sigma$ or abstains in the vote between them.
\end{lemma}

\begin{proof}
Let $\tau_{i}$ be a ranking of the candidates in $C_{i}$. We denote by $\tau_{1}\tau_{2}\ldots\tau_{\dmrev{k}}$ the ranking of the candidates $\cup_{i=1}^{\dmrev{k}} C_{i} = C$, in which each candidate in $C_{i}$ is preferred to each candidate in $C_{j}$ whenever $i<j$, and candidates within a set $C_{i}$ are ranked according to~$\tau_{i}$.
Let $\zeta_{i}$ be the ranking on $C_{i}$ that orders candidates in $C_{i}$ according to their rank in $\sigma$, that is, $\rank_{\zeta_{i}}[a]<\rank_{\zeta_{i}}[b]$ if $\rank_{\sigma}[a]<\rank_{\sigma}[b]$. 
Let $\zeta := \zeta_{1}\zeta_{2}\ldots\zeta_{k}$. So $\zeta$ preserves $(C_{1},\ldots, C_{k})$. Consider voter $v \in V$. If $v$ prefers $a$ to~$b$, and $\zeta$ orders $b$ before~$a$, then since by assumption $\pi_v$ preserves $(C_{1},\ldots, C_{k})$, it follows that $a,b \in C_{i}$ for some $1\leq i \leq m$. Since the relative order of candidates in $C_{i}$ is the same in $\zeta$ and $\sigma$ by construction, $\sigma$ also orders $b$ before~$a$. In particular, every pair of candidates that contributes 1 to $K(\zeta,\pi_{v})$ also contributes 1 to $K(\sigma,\pi_{v})$. We conclude for any $v\in V$, it holds that $K(\zeta,\pi_{v})\leq K(\sigma,\pi_{v})$, as desired.
\qed\end{proof}
\begin{theorem}\label{example}
There exists a profile with six voters that has a weakly popular ranking which is not strongly popular.
\end{theorem}

\begin{proof}
We prove this statement for the profile $P$ from Figure~\ref{fi:ex}.
\begin{myclaim}\label{cl1}
$\sigma_2 = [1,2,3],[4,5,6],[7,8,9]$ is weakly popular.
\end{myclaim}
\begin{proof}
Let $C_{1}=\{1,2,3\},C_{2}=\{4,5,6\},C_{3}=\{7,8,9\}$. Note that all voters preserve $(C_{1},C_{2},C_{3})$. By Lemma~\ref{pres}, in order to check if $\sigma_2$ is weakly popular it suffices to generate the $6^{3}$ rankings that preserve $(C_{1},C_{2},C_{3})$, and compare each of them to~$\sigma_2$. We checked all of these rankings using program code, which is available \dmrev{from} \url{https://github.com/SonjaKrai/PopularRankingsExampleCheck}.
\myqed
\end{proof}

\begin{myclaim}\label{cl2}
$\sigma_2 = [1,2,3],[4,5,6],[7,8,9]$ is not strongly popular.
\end{myclaim}

\begin{proof}
Ranking $\sigma_{1}=[1,2,3],[5,6,4],\allowbreak[9,7,8]$ is more popular than $\sigma_2$ in the simple sense, because voters $v_1, v_2, v_3$ abstain, voter $v_4$ prefers $\sigma_2$ to~$\sigma_1$, and finally, voters $v_5$ and $v_6$ prefer $\sigma_1$ to~$\sigma_2$. \ref{app:preferences} contains the detailed calculations \dmrev{that justify the preferences of each voter between $\sigma_1$ and~$\sigma_2$.}
\myqed\end{proof}
This finishes the proof of our theorem.\qed\end{proof}

\dmrev{We remark that Theorem \ref{example} can be easily extended to $k>6$ voters, for even values of $k$, by adding $k-6$ voters to the profile $P$ from Figure~\ref{fi:ex}, where half of the new voters have $\sigma_1$ as their ranking, and the other half of the new voters have $\sigma_2$ as their ranking.}

For the sake of completeness, we \dmrev{also} remark that the profile $P$ from Figure~\ref{fi:ex} has four further weakly popular rankings, each of which is strongly popular as well. For a profile that admits a weakly popular ranking but does not admit a strongly popular ranking, \dmrev{we refer to reader to} \ref{app:nosr}.

\subsection{\dmrev{When are weakly and strongly popular rankings equivalent?}}
\label{se:condition}

We now present a sufficient condition under which strong popularity follows from weak popularity. 
We start by showing that for a ranking $\pi$ that is not strongly popular, there is a condition under which we can compute a ranking that is also more popular than $\pi$ in the absolute sense.
    \begin{lemma}\label{condeq}
     If $\sigma_1$ is more popular than $\pi$ in the simple sense and the majority graph of the voters in $V_{\text{abs}}(\sigma_1, \pi)$ is acyclic, then there is a ranking $\sigma_2$ that is more popular than $\pi$ in the absolute sense. Such a $\sigma_2$ can be computed in polynomial time. 
    \end{lemma}
     \begin{proof}
 If $\pi$ is not a topologically sorted ranking then $\pi$ cannot be weakly popular by Observation~\ref{obs:topsort}.  
 As in the proof of Lemma~2 in \cite{VSW14}, we may construct a ranking $\sigma_2$ that is more popular than $\pi$ in the absolute sense. 
 In particular, there must be $a,b$ such that $rank_{\pi}[b]<rank_{\pi}[a]$ and \dmrev{an absolute} majority of voters in $V$ prefer $a$ to $b$. Let $\sigma_2$ be the ranking obtained by swapping $a$ and $b$. The proof of Lemma~2 in ~\cite{VSW14} shows that $\sigma_2$ is preferred to $\pi$ by \dmrev{an absolute} majority of voters who prefer $a$ to $b$. This shows that $\sigma_2$ is more popular in the absolute sense.
 Hence suppose that $\pi$ is a topologically sorted ranking.         
    Furthermore, if $V_{\text{abs}}(\sigma_1, \pi)=\emptyset$, then $\sigma_1$ is preferred to $\pi$ by \dmrev{an absolute} majority of all voters, and thus the theorem is proved by taking $\sigma_2 := \sigma_1$.
     
     From here on we therefore assume that $\pi$ is topologically sorted and that $V_{\text{abs}}(\sigma_1, \pi) \neq \emptyset$. Let $D(\sigma_1, \pi)\neq \emptyset$ be the set of pairs of candidates that are ordered differently in $\sigma_1$ and $\pi$.
     
     \begin{myclaim}
     \label{cl:Ceven}
        If $v \in V_{\text{abs}}(\sigma_1, \pi)$, then $v$ agrees with $\pi$ on the order of exactly half of the pairs in $D(\sigma_1, \pi)$, and disagrees on the other half. In particular, $|D(\sigma_1, \pi)|$ is even.
     \end{myclaim}
     
     \begin{proof}
     Any pair of candidates $\{a,b\}$ such that $\sigma_1$ and $\pi$ agree on the order $(a,b)$, adds $1$ to both $K(\pi_{v},\sigma_1)$ and $K(\pi_{v},\pi)$ if $\pi_{v}$ has them in the relative order $(b,a)$, and adds 0 to both otherwise. So to be impartial, i.e.\ to have $K(\pi_v,\sigma_1)=K(\pi_v,\pi)$, voter $v \in V_{\text{abs}}(\sigma_1, \pi)$ must agree with $\sigma_1$ on the order of exactly half of the pairs in $D(\sigma_1, \pi)$, and disagree on the other half\dmrev{. So $|D(\sigma_1, \pi)|$ must be even.}
    \myqed\end{proof}

     \begin{myclaim}
     \label{cl:ab}
     There exist consecutive candidates $(a^*,b^*)$ in $\sigma_1$ such that at least half of the voters in $V_{\text{abs}}(\sigma_1, \pi)$ prefer $b^*$ to~$a^*$.
     \end{myclaim}
     
     \begin{proof}
     Assume the contrary, i.e.\ that for any two consecutive candidates in $\sigma_1$, \dmrev{an absolute} majority of voters in $V_{\text{abs}}(\sigma_1, \pi)$ agree with~$\sigma_1$. For an arbitrary pair of candidates $\{a,b\}$, at least half of the voters in $V_{\text{abs}}(\sigma_1, \pi)$ must then also agree with their order in~$\sigma_1$, as otherwise this would imply a directed cycle in their majority graph, which is acyclic by assumption. 
     
     Since $\sigma_1\neq\pi$, there is some ordered pair $(a^*,b^*)$ that is consecutive in $\sigma_1$ and $b^*$ is ordered somewhere before $a^*$ in~$\pi$, that is $\rank_{\pi}[b^*]<\rank_{\pi}[a^*]$. Note that $(a^*,b^*) \in D(\sigma_1, \pi)$. In particular, by our starting assumption, \dmrev{an absolute} majority of voters in $V_{\text{abs}}(\sigma_1, \pi)$ agrees with $\sigma_1$ on the order $(a^*,b^*)$ and hence disagrees with~$\pi$.
     
     We now introduce the indicator variable $I_{v,\pi,\{a,b\}}$, which is  set to 1 if $\pi_{v}$ and $\pi$ disagree on the order of candidates $a$ and $b$, and it is set to 0 otherwise. We sum up the disagreements of voters in $V_{\text{abs}}(\sigma_1, \pi)$ with $\pi$ over pairs in $D(\sigma_1, \pi)$ and obtain a contradiction.
     
    \begin{eqnarray}
	\label{eq:2} |V_{\text{abs}}(\sigma_1, \pi)|\frac{|D(\sigma_1, \pi)|}{2}  & = & \sum_{\{a,b\} \in D(\sigma_1, \pi)} \sum_{v \in V_{\text{abs}}(\sigma_1, \pi)}  I_{v,\pi,\{a,b\}} \\ 
	\label{eq:3}  & > & \sum_{\{a,b\} \in D(\sigma_1, \pi)} \frac{|V_{\text{abs}}(\sigma_1, \pi)|}{2} \\
    \label{eq:3a} & = & |V_{\text{abs}}(\sigma_1, \pi)|\frac{|D(\sigma_1, \pi)|}{2} 
    \end{eqnarray}	
The right-hand side of Line~\ref{eq:2} is a formulation of disagreements in terms of the indicator variable. Due to Claim~\ref{cl:Ceven}, the number of disagreements that abstaining voters have with $\pi$ is exactly half of $|D(\sigma_1, \pi)|$, expressed on the left-hand side of Line~\ref{eq:2}. The inequality in Line~\ref{eq:3} follows because for all pairs in $D(\sigma_1, \pi)$, at least half of the abstaining voters disagree with $\pi$, and, additionally, there exists a pair $\{a^*,b^*\}$ such that \dmrev{an absolute} majority of voters in $V_{\text{abs}}(\sigma_1, \pi)$ disagree with~$\pi$, as we proved above. Finally, reordering the terms as in Line~\ref{eq:3a} leads back to the same formula as on the left-hand side in Line~\ref{eq:2}, creating a contradiction.
\myqed\end{proof}
     
     The pair of candidates $(a^*,b^*)$ in Claim~\ref{cl:ab} leads us to a suitable ranking~$\sigma_2$. Let $\sigma_2$ be the ranking we get from $\sigma_1$ by the swap $(a^*,b^*) \rightarrow (b^*,a^*)$. We now prove Claims~\ref{cl:preferb} and~\ref{cl:preferpi1}, which  we will use to show that two groups of voters prefer  $\sigma_2$ to $\pi$, and that these two groups constitute \dmrev{an absolute} majority of all voters.
     
     \begin{myclaim}
     \label{cl:preferb}
        All voters $v \in V_{\text{abs}}(\sigma_1, \pi)$ who prefer $b^*$ to $a^*$ also prefer $\sigma_2$ to $\pi$.
     \end{myclaim}
     \begin{proof}
        If $v \in V_{\text{abs}}(\sigma_1, \pi)$ prefers $b^*$ to $a^*$, then the following holds.
    \begin{eqnarray}
	\label{eq:4} K(\pi_v,\sigma_2) & = & K(\pi_v,\sigma_1)-1 \\ 
	\label{eq:5} & < & K(\pi_v,\sigma_1) \\
    \label{eq:5a} & = & K(\pi_v,\pi)
	\end{eqnarray}
	Line~\ref{eq:4} is true, because $\sigma_2$ is obtained from $\sigma_1$ by performing a swap that is good for voter~$v$. The equality in Line~\ref{eq:5a} holds since $v$ abstains between $\sigma_1$ and~$\pi$. Line~\ref{eq:5a} together with the inequality in Line~\ref{eq:5} prove the claim.
     \myqed\end{proof}
     
     The second group of voters we investigate consists of voters who prefer $\sigma_1$ to $\pi$. This group by assumption makes up \dmrev{an absolute} majority of the non-abstaining voters $V \setminus V_{\text{abs}}(\sigma_1, \pi)$. Let voter $v$ belong to this group.
       
     \begin{myclaim}
     \label{cl:preferpi1}
        If $K(\pi_{v},\sigma_1)<K(\pi_{v},\pi)$ for voter $v \in V \setminus V_{\text{abs}}(\sigma_1, \pi)$, then $K(\pi_{v},\sigma_1)+2\leq K(\pi_{v},\pi)$.
     \end{myclaim}
     \begin{proof}
     Only the pairs in $D(\sigma_1, \pi)$ contribute differently to $K(\pi_{v},\sigma_1)$ and to $K(\pi_{v},\pi)$. In particular, a pair in $D(\sigma_1, \pi)$ adds $1$ to either $K(\pi_{v},\sigma_1)$ or $K(\pi_{v},\pi)$, and $0$ to the other. Using this we can see that if $k$ is the number of pairs on whose order $\sigma_1$ and $\pi$ agree, but $\pi_{v}$ disagrees, then $K(\pi_{v},\sigma_1)+K(\pi_{v},\pi)=|D(\sigma_1, \pi)|+2k$, which is even by Claim~\ref{cl:Ceven}. Since $K(\pi_{v},\sigma_1)<K(\pi_{v},\pi)$ by assumption and their sum $|D(\sigma_1, \pi)|+2k$ is even, $K(\pi_{v},\sigma_1)+2\leq K(\pi_{v},\pi)$ must hold.
     \myqed\end{proof}
     Since a swap of consecutive candidates in a ranking can increase the distance to any other ranking by at most one, Claim~\ref{cl:preferpi1} implies that for any voter $v$ who prefers $\sigma_1$ to $\pi$, the following holds:
     $$K(\pi_{v},\sigma_2)\leq  K(\pi_{v},\sigma_1)+1<K(\pi_{v},\pi).$$
     We conclude that voters who prefer $\sigma_1$ to $\pi$ also prefer $\sigma_2$ to~$\pi$.
     
    Claims~\ref{cl:ab} and~\ref{cl:preferb} \dmrev{imply} that at least half of the voters in $V_{\text{abs}}(\sigma_1, \pi)$ prefer $\sigma_2$ to $\pi$, and Claim~\ref{cl:preferpi1} proves that more than half of the voters outside of $V_{\text{abs}}(\sigma_1, \pi)$ prefer $\sigma_2$ to~$\pi$. The two sets of voters thus constitute \dmrev{an absolute} majority of all voters who prefer $\sigma_2$ to~$\pi$. 
     \qed\end{proof}

By rephrasing \dmrev{Lemma}~\ref{condeq}, we \dmrev{obtain} the following.

\begin{theorem}\label{cor:equiv}
If ranking $\pi$ is weakly popular, and for any ranking $\sigma$, $V_{\text{abs}}(\sigma, \pi)$ has an acyclic majority graph, then $\pi$ is also strongly popular. Furthermore, if $V_{\text{abs}}(\sigma, \pi)$ has an acyclic majority graph for each weakly popular ranking $\pi$ and any ranking $\sigma$ for a profile $P$, then weakly and strongly popular rankings for $P$ coincide.
\end{theorem}
\begin{observation}
\label{obs:nec}
\dmrev{
The acyclicity condition for the majority graph of the abstaining voters in Lemma~\ref{condeq} and Theorem~\ref{cor:equiv} is necessary.}
\end{observation}

\begin{proof}
\dmrev{
 We construct an example to demonstrate that there is a profile and rankings $\sigma_1$, $\sigma_2$ such that the majority graph of the voters in $V_{\text{abs}}(\sigma_1, \sigma_2)$ is cyclic (contradicting our assumption in Lemma~\ref{condeq} and Theorem~\ref{cor:equiv}), and $\sigma_2$ is weakly popular but not strongly popular.}

\dmrev{
Consider the profile described in Figure~\ref{fi:ex} with $\sigma_{1}=[1,2,3],[5,6,4],[9,7,8]$ and $\sigma_{2}=[1,2,3], [4,5,6],[7,8,9]$, as in the proof of Theorem~\ref{example}. Notice that there are three directed cycles in the majority graph of the three abstaining voters $v_1,v_2,v_3$: one for candidates $1,2,3$, one for candidates $4,5,6$ and one for candidates $7,8,9$. So the abstaining voters have a cyclic majority graph.}

\dmrev{As shown by Claims \ref{cl1} and \ref{cl2} in the proof of Theorem~\ref{example}, $\sigma_2$ is 
weakly popular but not strongly popular, since $\sigma_1$ is more popular than $\sigma_2$ in the simple sense. Therefore, this example justifies the necessity of the acyclicity condition for the majority graph of the abstaining voters in Lemma~\ref{condeq} and Theorem~\ref{cor:equiv}.
}\qed\end{proof}

\subsection{At most \dmrev{five} voters}
\label{se:5voters}

From \dmrev{Lemma}~\ref{condeq} we can deduce that for a small number of voters, the two notions of popularity are equivalent. This is due to the fact that we need at least \dmrev{three} abstaining voters in order for $V_{\text{abs}}(\sigma, \pi)$ to have a cyclic majority graph.

\begin{theorem}\label{atmost5}
A ranking $\sigma$ is weakly popular for a profile $P$ with at most five voters if and only if it is strongly popular for~$P$.
\end{theorem}
\begin{proof}
From Observation~\ref{obs:topsort} we know that strongly popular rankings are also weakly popular, which allows us to concentrate only on one direction of the statement, namely that if a ranking is not strongly popular, then it also cannot be weakly popular. Let us assume that ranking $\pi$ is not strongly popular. By definition there exists a ranking $\sigma$ that is preferred to $\pi$ by \dmrev{an absolute} majority of non-abstaining voters $V \setminus V_{\text{abs}}(\sigma, \pi) \neq \emptyset$.

If at least one voter prefers $\pi$ to $\sigma$, then at least two voters must prefer $\sigma$ to $\pi$, and the remaining at most two abstaining voters can only form an acyclic majority graph. From this it follows by \dmrev{Lemma}~\ref{condeq} that $\pi$ is not weakly popular.

On the other hand, if no voter prefers $\pi$ to $\sigma$, then $K(\pi_{v},\sigma) \leq K(\pi_v,\pi)$ holds for all voters. For $V \setminus V_{\text{abs}}(\sigma, \pi) \neq \emptyset$, by assumption there is a voter $v^*$ who prefers $\sigma$ to $\pi$, that is, for whom $K(\pi_{v^*},\sigma)<K(\pi_{v^*},\pi)$. We thus have $\sum_{i=1}^{n} K(\pi_{v_{i}},\sigma)< \allowbreak \sum_{i=1}^{n} K(\pi_{v_{i}},\pi)$. This means that $\pi$ is not a Kemeny consensus and by Observation~\ref{obs:topsort}, $\pi$ is not weakly popular.
\qed\end{proof}

\section{On the complexity of \dmrev{$k$-{\sc wurv}} and $k$-{\sc surv}}
\label{smallvoters}

In this section, we analyse the complexity of the verification problems for the two popularity notions. We prepare for this by giving a sufficient condition for weak popularity in Section~\ref{sec:suffcond}. This condition will then be used in Section~\ref{sec:polsolv}, where we prove the polynomial solvability of both problems in the case of $k \leq 3$. For $4\leq k \leq 5$, we reach the same conclusion in Section~\ref{sec:45}, however, only for special profiles. Then, in Section~\ref{sec:n6}, $\NP$-hardness is proved for $k=6$.

\subsection{A sufficient condition for \dmrev{{\sc wurv}}}
\label{sec:suffcond}

We call a ranking $\pi$ \textit{$c$-sorted} for $0 < c \leq 1$ if for all pairs of candidates $\{a,b\}$ with $\rank_{\pi}[a]<\rank_{\pi}[b]$, at least a $c$-fraction of the voters prefers $a$ to~$b$. A ranking $\pi$ \dmrev{is topologically sorted} \dmrev{if and only if} it is $\frac{1}{2}$-sorted. In \cite{VSW14}, van Zuylen et al.\ show that even a topologically sorted ranking is not necessarily weakly popular. Here we ask for which constant $\frac{1}{2} < c\leq 1$ does this negative result change to a positive result, guaranteeing a \dmrev{no} answer for \dmrev{$k$-{\sc wurv}}.  \dmrev{(Note that we do not consider the case where $c<\frac{1}{2}$, since any $c$-popular ranking cannot be topologically sorted and hence cannot be weakly popular by Observation 
\ref{obs:topsort}.)}

\begin{theorem}\label{tightc}
$c=\frac{3}{4}$ is the smallest constant (\dmrev{$\frac{1}{2}<c\leq 1$})
for which the following holds: If a profile $P$ admits a $c$-sorted ranking $\pi$, then $\pi$ is weakly popular.
\end{theorem}
\begin{proof}
 We first show that any $\frac{3}{4}$-sorted ranking is weakly popular. 
\begin{myclaim}\label{cl:34wp}
\dmrev{Let $P$ be a profile and $\pi$ be a $\frac{3}{4}$-sorted ranking for $P$.  Then $\pi$ is weakly popular.}
\end{myclaim}

\begin{proof}
Let $\sigma$ be any ranking different from~$\pi$. We will show that there is no voter set of cardinality $\floor{\frac{n}{2}}+1$, in which every voter prefers $\sigma$ to~$\pi$. Let $V'$ be an arbitrary set of $\floor{\frac{n}{2}}+1$ voters. On the bubble swap path  $\pi_{0}:=\pi,\pi_{1},\ldots \pi_{k}:=\sigma$ from $\pi$ to $\sigma$, each swap $\pi_{i}\rightarrow \pi_{i+1}$ is bad for at least \dmrev{a} \dmrev{$\frac{3}{4}$\dmrev{-fraction of} voters in $V$. So it must be bad for at least $$\dmrev{\left\lfloor\frac{n}{2}\right\rfloor}+1-\frac{n}{4}\geq \frac{|V'|}{2}$$
voters in $V'$.}
Hence for a swap $\pi_{i}\rightarrow\pi_{i+1}$, we have that $\sum_{v\in V'}(K(\pi_{i+1},\pi_{v})-K(\pi_{i},\pi_{v})) \geq 0$. Summing over all swaps we get a telescoping sum that reduces to $$\sum_{v\in V'}(K(\sigma,\pi_{v})-K(\pi,\pi_{v}))\geq 0.$$ In particular, not every $v\in V'$ can prefer $\sigma$ to~$\pi$. Since $V'$ was an arbitrary set of size $\floor{\frac{n}{2}}+1$, no voter set of at least this size can prefer $\sigma$ to $\pi$, and thus, $\pi$ must be weakly popular.
\myqed\end{proof}

We now show that in fact $c=\frac{3}{4}$ is tight, meaning that for any \dmrev{$\frac{1}{2}<$} $c<\frac{3}{4}$, we can construct a profile $P$ and a $c$-sorted ranking $\pi$ such that $\pi$ is not weakly popular for~$P$.
\begin{myclaim}
\label{th:topsort_notwp}
 For arbitrary $\frac{1}{2}\leq c<\frac{3}{4}$ there exists a profile $P$ \dmrev{and a} \dmrev{c-sorted} \dmrev{ranking} \dmrev{$\pi$ \dmrev{that} is not weakly popular}.
\end{myclaim}

\begin{proof}
\dmrev{Let $c=\frac{3}{4}-\varepsilon$ for some $\frac{1}{4}\geq \varepsilon>0$. Next, choose $j$ such that $j \geq \frac{1}{4\varepsilon}$.} 
We \dmrev{will} create profile $P$ with $4j$ voters and $4j+2$ candidates. 
Then it holds that $(\frac{1}{4}+\varepsilon)|V|\geq j+1$.

\dmrev{Each voter's ranking involves $2j+1$ \emph{blocks}, where block $i$ comprises candidates $\{2i-1,2i\}$ ($1\leq i\leq 2j+1$).  We say that block $i$ is \emph{increasing} if it is in the form $[2i-1,2i]$ and it is \emph{decreasing} if it is in the form $[2i,2i-1]$.}

\dmrev{We firstly create a set $V_1$ of $2j-1$ voters, and each voter in $V_1$ has only increasing blocks, meaning that their ranking is $\pi=[1,2],[3,4],\ldots,$ $\mbox{[4$j$+1,4$j$+2]}$.
The \dmrev{set $V_2$ comprises the $2j+1$} remaining voters, and their set of rankings is constructed as follows. The $i$th voter in $V_2$ ($1\leq i\leq 2j+1$) has blocks $i\mod (2j+1)$ to $(i+j)\mod (2j+1)$ decreasing, whilst all other blocks are increasing.  This construction is illustrated in Figure~\ref{ta:blocks} for the case that $j=3$.
}
This construction is similar to the one in~\cite[Lemma~3]{VSW14}.

\begin{figure}[ht]
\begin{center}
\resizebox{0.65\columnwidth}{!}{%
\begin{tabular}{ L L L L L L L L}
 \cellcolor{C1}[2,1] & \cellcolor{C1}[4,3]  & \cellcolor{C1}[6,5] & \cellcolor{C1}[8,7] & [9,10] & [11,12] & [13,14]\\
 \cellcolor{white}[1,2] & \cellcolor{C1}[4,3] & \cellcolor{C1}[6,5] & \cellcolor{C1}[8,7] & \cellcolor{C1}[10,9] & [11,12] 
 & [13,14]\\
 \cellcolor{white}[1,2] & [3,4] & \cellcolor{C1}[6,5] & \cellcolor{C1}[8,7] & \cellcolor{C1}[10,9] & \cellcolor{C1}[12,11] & [13,14]\\
 \cellcolor{white}[1,2] & [3,4] & [5,6] & \cellcolor{C1}[8,7] & \cellcolor{C1}[10,9] & \cellcolor{C1}[12,11]  
& \cellcolor{C1}[14,13]\\
\cellcolor{C1}[2,1] & [3,4] & [5,6] & [7,8] & \cellcolor{C1}[10,9] & \cellcolor{C1}[12,11]  
& \cellcolor{C1}[14,13]\\
  \cellcolor{C1}[2,1] & \cellcolor{C1}[4,3] &[5,6] & [7,8] & [9,10] & \cellcolor{C1}[12,11]
  & \cellcolor{C1}[14,13]\\
  \cellcolor{C1}[2,1] & \cellcolor{C1}[4,3] &\cellcolor{C1}[6,5] & [7,8] & [9,10] & [11,12]
  & \cellcolor{C1}[14,13]
  

\end{tabular}
}
\end{center}
    \caption{\dmrev{The \dmrev{rankings of the voters $V_2$, as described in the proof of Claim \ref{th:topsort_notwp}, for the special case where $j=3$.}  The \dmrev{decreasing blocks} are highlighted.}}
    \label{ta:blocks}
\end{figure}

For each \dmrev{$\{2i-1,2i\}$} for $1\leq i \leq 2j+1$, every $v\in V_{1}$ agrees with this pair and also exactly $j$ of the voters in $V_{2}$ agree with it. In total, $3j-1$ voters thus agree with this pair.
 This gives the following fraction of all voters:
 $$\frac{3j-1}{4j}=\frac{3}{4}-\frac{1}{4j} \geq \frac{3}{4}- \varepsilon.$$
It is trivial to see for all other pairs of voters $\{a,b\}$ that if $a<b$ then all voters prefer $a$ to~$b$. Hence in particular $\pi$ is $\left(\frac{3}{4}-\varepsilon\right)$-sorted.

We now show that ranking $\sigma=[2,1],[4,3],\ldots [4j+2,4j+1]$ is preferred to $\pi$ by \dmrev{an absolute} majority of all voters, namely by all $2j+1$ voters in~$V_{2}$. Each voter in $V_{2}$ has exactly $j+1$ blocks with decreasing order and $j$ blocks with increasing order. \dmrev{Therefore}, each of them would rather have all pairs in decreasing order than all pairs in increasing order.
\myqed\end{proof}
\dmrev{Theorem\dmrev{~\ref{tightc}} is thus established.} 
\qed\end{proof}


As an aside, following on from Theorem~\ref{tightc}, \dmrev{it is natural to ask about the existence of 
\emph{$c$-popular rankings:} rankings that are preferred to all other rankings by some $c$-fraction of voters.} 


\begin{theorem}
\label{th:csorted}
\dmrev{There is a profile that does not admit a c-popular ranking for any $c>0$.}
\end{theorem}
\begin{proof}
\dmrev{\dmrev{Let $n=\ceil{\frac{1}{c}}+1$.} We exhibit a profile of $n$ voters over $n$ candidates, such that for any ranking $\pi$, there exists another ranking $\sigma$ such that $n-1$ voters prefer $\sigma$ to $\pi$ and only one voter prefers $\pi$ to~$\sigma$.}
\dmrev{This implies that $\pi$ cannot be preferred to any other ranking by a $c$-fraction of the voters, since $\frac{1}{n}<c$.}

Consider the extended Condorcet paradox, i.e.\ $n$ voters with rankings of $n$ candidates as follows:
\begin{center}
\begin{tabular}{rccl}
$\pi_{v_{1}}=$ & $1,2, $&$\ldots,$&$ n-1,n$\\ 
$\pi_{v_{2}}=$ & $2,3, $&$\ldots,$&$ n,1$\\ 
 \vdots&&&\\
$\pi_{v_{n}}$= &$ n,1,$ &$\ldots,$& $n-2,n-1$
\end{tabular}
\end{center}

Let $\pi$ be an arbitrary ranking of the $n$ candidates. 
For each ordered pair $(a,b)$ in the set $A:=\{(1,2),(2,3),\ldots (n-1,n),(n,1)\}$, exactly $n-1$ voters of the above instance agree with $(a,b)$.
We also know that there must be at least one ordered pair in $A$ with which \dmrev{the ranking $\pi$ disagrees.
Let $(a,b)$ be such a pair, so $\pi$ \dmrev{prefers $b$ to $a$}, but $n-1$ voters prefer $a$ to $b$.}

Let $\sigma$ be the ranking obtained from $\pi$ by swapping $b$ and~$a$. 
\dmrev{Let \dmrev{$d$} be a candidate ranked between $b$ and $a$ in~$\pi$. Each voter $v$ who prefers $a$ to $b$ must also prefer $a$ to \dmrev{$d$} or \dmrev{prefer $d$} to~$b$. So since $\rank_{\dmrev{\pi}}[b]<\rank_{\dmrev{\pi}}[\dmrev{d}]<\rank_{\dmrev{\pi}}[a]$ together the pairs $(a,\dmrev{d})$ and $(\dmrev{d},b)$ add at least \dmrev{$1$} to $K(\pi_v,\pi)$. 
Since $\rank_{\dmrev{\sigma}}[a]<\rank_{\dmrev{\sigma}}[\dmrev{d}]<\rank_{\dmrev{\sigma}}[b]$, at most one of the pairs $(a,\dmrev{d})$ and $(\dmrev{d},b)$ adds \dmrev{$1$} to $K(\pi_v,\sigma)$.}  Pair $(a,b)$ adds \dmrev{$1$} to $K(\pi_v,\pi)$ and $0$ to $K(\pi_v,\sigma)$.  From the definition of $\sigma$, it follows that $K(\pi_v,\sigma)<K(\pi_v,\pi)$ for each voter who prefers $a$ to $b$ and there are $n-1$ such voters $v$.
\qed\end{proof}

\subsection{Polynomial-time solvability for $k\leq 3$}
\label{sec:polsolv}

Since we have shown in Theorem~\ref{atmost5} that \dmrev{weak} and \dmrev{strong} popularity are equivalent notions for $k \leq 5$, we will refer to them as 
popularity if $k \leq 5$.

We first show that for $k\leq 3$, the problems $k$-\textsc{wurv} and $k$-\textsc{surv} are polynomial-time solvable. We establish this by proving that for at most three voters, the set of topologically sorted rankings coincides with the set of popular rankings. Since verifying whether a given ranking is topologically sorted can be carried out in polynomial time, \dmrev{$k$-\textsc{wurv}} and $k$-\textsc{surv} turn out to be polynomial-time solvable for $k=2,3$.
\begin{lemma}
     \label{cl:twovoters}
       Given a profile of two voters, a ranking is popular if and only if it is topologically sorted.
     \end{lemma}
\begin{proof}
From Observation~\ref{obs:topsort} we know that all popular rankings must be topologically sorted. Let $D(\pi_{v_{1}},\pi_{v_{2}})$ be the set of pairs of candidates $\{a,b\}$ that $\pi_{v_{1}}$ and $\pi_{v_{2}}$ order differently. Consider any ranking $\pi$. Each pair $\{a,b\} \in D(\pi_{v_{1}},\pi_{v_{2}})$ adds $1$ to either $K(\pi_{v_{1}},\pi)$ or $K(\pi_{v_{2}},\pi)$. 
From this follows that
$$|D(\pi_{v_{1}},\pi_{v_{2}})|\leq K(\pi_{v_{1}},\pi)+K(\pi_{v_{2}},\pi).$$

If $\sigma$ is a topologically sorted ranking, then by definition there is no pair of candidates $\{a,b\}$ that adds $1$ to both $K(\pi_{v_{1}},\sigma)$ and $K(\pi_{v_{2}},\sigma)$. Thus,  $K(\pi_{v_{1}},\sigma)+K(\pi_{v_{2}},\sigma)= |D(\pi_{v_{1}},\pi_{v_{2}})|$. If $\pi$ is preferred to $\sigma$ by \dmrev{an absolute} majority, then both voters prefer $\pi$ to $\sigma$, which leads to
\begin{eqnarray*}
|D(\pi_{v_{1}},\pi_{v_{2}})| & \leq & K(\pi_{v_{1}},\pi)+K(\pi_{v_{2}},\pi)\\
& < & K(\pi_{v_{1}},\sigma)+K(\pi_{v_{2}},\sigma)\\
& = & |D(\pi_{v_{1}},\pi_{v_{2}})|.
\end{eqnarray*}
Since this is a contradiction, $\sigma$ must be weakly popular. By Theorem~\ref{atmost5}, $\sigma$ is also strongly popular.
     \qed\end{proof}

\begin{lemma}
     \label{3voters}
        Given a profile $P$ of three voters, a ranking is popular if and only if it is topologically sorted.
     \end{lemma}
     
\begin{proof}
Just as for the $k=2$ case, Observation~\ref{obs:topsort} implies here as well that all popular rankings must be topologically sorted. Let $\pi$ be a topologically sorted ranking for~$P$. Note that whenever $\rank_{\pi}[a]~<~\rank_{\pi}[b]$ holds for candidates $a$ and $b$, at least half of the voters, that is, at least two voters prefer $a$ to $b$. So for any two voters, at least one of them prefers $a$ to $b$, implying that $\pi$ is also topologically sorted for any two of the three voters. In particular, $\pi$ is weakly popular for any two voters by Lemma~\ref{cl:twovoters}, showing that there is no ranking that they both prefer to~$\pi$. Hence $\pi$ must be weakly popular \dmrev{for} $P$ and by Theorem~\ref{atmost5}, also strongly popular.
\qed\end{proof}

\dmrev{Lemmas \ref{cl:twovoters} and \ref{3voters} lead to the following result regarding the complexity of $k$-\textsc{wurv} and $k$-\textsc{surv} for $k\leq 3$, and the complexity of finding a popular ranking or reporting that none exists in the case that $k\leq 3$.}
\begin{theorem}\label{th:smallk}
\dmrev{For $k\leq 3$, $k$-\textsc{wurv} and $k$-\textsc{surv} are solvable in $O(m^2n)$ time.  Moreover for $k\leq 3$, we can find a popular ranking or report that none exists in $O(m^2n)$ time.}
\end{theorem}
\begin{proof}
\dmrev{Let $D$ denote the majority graph for the given profile $P$.  Lemmas \ref{cl:twovoters} and \ref{3voters} state that a given ranking is popular if and only if it is topologically sorted for $P$.  Moreover a topologically sorted ranking exists if and only if $D$ is acyclic.}

\dmrev{To establish the time complexity, clearly it is trivial to construct in $O(m)$ time for each voter $v_i$ a data structure that allows us to check in $O(1)$ time whether $v_i$ prefers $a$ to $b$, for any pair of candidates $a,b$.  Thus the $m\times m$ matrix $N$, where $N(a,b)$ gives the number of voters who prefer candidate $a$ to candidate~$b$ (i.e.\ $N(a,b)=|V(a,b)|$), can be constructed in $O(m^2n)$ time.}

\dmrev{Using $N$, we can then compute $D$ in $O(m^2)$ time.  To find a popular ranking or report that none exists, we can check in $O(m^2)$ time whether $D$ is acyclic, and if so, construct a topological ordering of $D$ in the same time complexity. Alternatively, for a given ranking $\pi$, clearly we can check whether $\pi$ is topologically sorted for $P$ in $O(m^2)$ time.  The overall time complexity is thus $O(m^2n)$ for both the verification and search algorithms.}
\qed 
\end{proof}

\subsection{Equivalence of the cases $k=4$ and $k=5$}
\label{sec:45}
If we have \dmrev{four} or \dmrev{five} voters, it turns out a topologically sorted ranking may not be popular anymore. In the case of an acyclic tournament as the majority graph, finding and verifying a popular ranking are both polynomial-time solvable.
We further show that  $4$-\textsc{wurv} ($4$-\textsc{surv}) in general is polynomial-time solvable if and only if $5$-\textsc{wurv} ($5$-\textsc{surv}) is.

\begin{lemma}
     \label{cl:fourvotersa}
        If a profile $P$ of four voters has an acyclic tournament as its majority graph, then the unique topologically sorted ranking is the unique popular ranking.
     \end{lemma}  
\begin{proof}Since the majority graph of $P$ is a tournament, the unique topologically sorted ranking $\pi$ is $\frac{3}{4}$-sorted.
The lemma then follows from Theorem~\ref{tightc} and Observation~\ref{obs:topsort}.
\qed\end{proof}

\begin{lemma}
     \label{cl:4votersc}
        \dmrev{Let $P$ be a profile of four voters with an acyclic majority graph, and let $\pi$ be a ranking for $P$ that is popular for at least one of the profiles formed by three of the voters.  Then $\pi$ is popular for~$P$.}
     \end{lemma}
\begin{proof}
Let $\pi$ be popular for the profile $(\pi_{v_{1}},\pi_{v_{2}},\pi_{v_{3}})$. Then by definition there exists no ranking $\sigma$ that is preferred by two of $v_{1},v_{2},v_{3}$, as this would contradict the fact that $\pi$ is weakly popular for the profile $(\pi_{v_{1}},\pi_{v_{2}},\pi_{v_{3}})$ by Theorem~\ref{atmost5}. In particular, there does not exist a ranking $\sigma$ preferred by \dmrev{an absolute} majority of $v_{1},v_{2},v_{3},v_{4}$, as this would require at least \dmrev{three} voters and hence at least $2$ of $v_{1},v_{2},v_{3}$. We conclude that $\pi$ is weakly popular and hence popular for~$P$.
\qed\end{proof}

We now present an example in which there is ranking $\pi$ that is not topologically sorted such that $\pi$ is more popular than a topologically sorted ranking.

\begin{observation}
     \label{cl:fourvotersb}
        Given a profile of four voters with an acyclic majority graph, a ranking that is not topologically sorted can be more popular than a topologically sorted ranking.
     \end{observation}     
\begin{proof}
Consider the following profile with $n=4$ and $m=10$.
\begin{align*}
\pi_{v_{1}} & =[1,2],[3,4],[5,6],[7,8],[9,10] \\ 
\pi_{v_{2}} & =[1,2],[4,3],[6,5],[7,8],[10,9] \\
\pi_{v_{3}} & =[1,2],[4,3],[6,5],[8,7],[9,10] \\
\pi_{v_{4}} & =[2,1],[3,4],[5,6],[8,7],[10,9]
\end{align*}
    It is easy to verify that $\pi_{v_{1}}$ is a topologically sorted ranking of~$P$. However, $\sigma=[2,1],[4,3],[6,5],[8,7],[10,9]$ is preferred by $v_{2}$, $v_{3}$, and $v_{4}$ to $\pi_{v_1}$,
    since $K(\pi_{v_{i}},\pi_{v_{1}})=3$ and  $K(\pi_{v_{i}},\sigma)=2$ for $2\leq i \leq 4$.
    Since \dmrev{an absolute} majority of voters prefer candidate~$1$ to candidate~$2$, $\sigma$ is not topologically sorted. For the sake of completeness, we remark that the topologically sorted ranking $[1,2],[4,3],[6,5],[8,7],[10,9]$ is popular.
    \qed\end{proof}   

We now discuss \dmrev{a family of strongly related decision problem\dmrev{s} called $k$-\textsc{all-closer-ranking}}, which will come useful when establishing results for the cases $k=4$ and $k=5$. For a \dmrev{profile} with $k$ voters and a given ranking $\pi$, we ask whether there is a ranking that all the voters prefer to~$\pi$.
\begin{tcolorbox}

\textbf{$k$-\fauxsc{all-closer-ranking}}\\
\textbf{Input}:  A set $C$, a profile $P=(\pi_{v_1},\ldots,\pi_{v_k})$ over $C$ and a ranking $\pi$ over $C$.\\
\textbf{Question}: Does there exist a ranking $\sigma$ that is preferred to $\pi$ by each of the $k$ voters?
\end{tcolorbox}

The next theorem reveals some features of this problem.
\begin{theorem}
\label{th:3acr}
Given a profile $P=(\pi_{v_1},\pi_{v_2},\pi_{v_3})$ with an acyclic majority graph, we can decide in polynomial time if there exists a ranking preferred by all voters to a given ranking~$\pi$ and if it does, output it.
\end{theorem} 
\begin{proof}
\dmrev{We start with two technical observations that will come \dmrev{in} handy later in our proof.}

\begin{observation}\label{obschange}
If $K(\sigma,\dmrev{\pi_{v_i}})>0$| for a voter $\dmrev{v_i, 1 \leq i \leq 3}$ and a ranking $\sigma$, then there exists a swap in $\sigma$ that is good for~\dmrev{$v_i$}.
\end{observation}

\begin{proof}
Suppose there is no swap in $\sigma$ that is good for~\dmrev{$v_i$}. So $\sigma$ is a topologically sorted ranking for \dmrev{$v_i$}, and for one voter this means $\dmrev{\pi_{v_i}}=\sigma$, i.e.\ $K(\sigma,\dmrev{\pi_{v_i}})=0$.
\myqed\end{proof}

\begin{observation}\label{obstwo}
\dmrev{Let $\pi$ be a ranking such that there is no swap in $\pi$ that is good for both $v_{i}$ and $v_{j}$, where $1 \leq i,j \leq 3$.
}
\dmrev{Then there is no ranking $\sigma$ preferred to $\pi$ by both $v_{i}$ and $v_{j}$.}
\end{observation}

\begin{proof}
Since there is no swap in $\pi$ that is good for both \dmrev{$v_i$} and \dmrev{$v_j$}, $\pi$ is a topologically sorted ranking for \dmrev{$v_i$} and \dmrev{$v_j$}. By Lemma~\ref{cl:twovoters} this means that $\pi$ is weakly popular \dmrev{for the sub-profile comprising $v_i$ and $v_j$}, 
so there exists no ranking preferred by \dmrev{an absolute} majority, that is, preferred by both voters.
\myqed\end{proof}

We are now ready to proceed to the \dmrev{main part of the proof.} 
\dmrev{First, note that we can check in polynomial time whether $\pi$ is topologically sorted for any two of the voters and hence by Lemma \ref{cl:twovoters} whether there is a ranking that they both prefer. Clearly if \dmrev{there is a pair of voters such that no ranking exists that is preferred to $\pi$ by both voters, }
then there is no ranking preferred  to $\pi$ by all three voters. So we may assume that for any two of the three voters there is a ranking they both prefer to~$\pi$.}

Second, note that since the number of voters is odd, the majority graph is a tournament. First we compute the unique topologically sorted ranking $\sigma$ of $P$ in polynomial time~\cite{Kah62}. We distinguish four cases, based on how many of the three voters prefer $\sigma$ to $\pi$, which can be checked in polynomial time.

\dmrev{\textbf{Case 1}}: If $\sigma$ is preferred to $\pi$ by all \dmrev{three} voters, then we are done.

\dmrev{\textbf{Case 2}}: Suppose that two of the voters, \dmrev{without loss of generality} $v_{1}$ and $v_{2}$, prefer $\sigma$ to $\pi$. \dmrev{Let $d_i=K(\pi,\pi_{v_i})-K(\sigma,\pi_{v_i})$ for $i\in \{1,2\}$.  Then $d_i\geq 1$ for $i\in \{1,2\}$.  Without loss of generality assume that $d_1\leq d_2$.  Also let $d_{3}=K(\sigma,\pi_{v_{3}})-K(\pi,\pi_{v_{3}})$.  Then $d_3\geq 0$.}
Let $\pi_{0}:=\pi,\ldots, \pi_{k}:=\sigma$ be the bubble sort swap path from $\pi$ to $\sigma$. 
    
    \begin{myclaim}\label{claim9}
    For the above defined distances, $d_{1}-d_{3} \geq 2$ and similarly, $d_{2}-d_{3} \geq 2$ hold.
    \end{myclaim}
    \begin{proof}
        Firstly, no swap in the bubble sort path is bad for both $v_{1}$ and $v_{3}$, since every swap $\pi_{i}\rightarrow\pi_{i+1}$ that is bad for $v_{3}$ must be good for both $v_{2}$ and $v_{1}$, because $\sigma$ is the topologically sorted ranking of~$P$.  If there is no swap that is good for both $v_{1}$ and $v_{3}$, then there cannot exist a ranking preferred to $\pi$ by all voters, since there cannot exist a ranking preferred by both $v_{1}$ and $v_{3}$ by Observation~\ref{obstwo}. So there is at least one swap in the bubble sort path that is good for both $v_{1}$ and~$v_{3}$. This swap adds $1$ to $d_1$ and subtracts $1$ from $d_3$, i.e.\ it adds $1$ to $-d_3$. By the previous argument, any other swap is good for at least one of $v_1$ and $v_3$, adding at least $0$ to $d_{1}-d_{3}$. 
        It follows that $d_{1}-d_{3} \geq 2$ and since $d_{2}\geq d_{1}$, it follows that the same argument also implies $d_{2}-d_{3} \geq 2$.
    \myqed\end{proof} 
    We now show how to transform $\sigma$ to a ranking that is preferred by all three voters to~$\pi$ if and only if such a ranking exists.\\
    
    \noindent \textbf{Procedure}\\
    Let $\sigma_{0}=\sigma$. In the $i$th round we search for a swap in $\sigma_{i-1}$ that is good for $v_{3}$ and if found, perform the swap to obtain $\sigma_{i}$ for $i \geq 1$. Otherwise we output an error message. We stop the procedure in round $i=d_{3}+1$ and output~$\sigma_{d_{3}+1}$.
    
    \begin{myclaim}
        If the \dmrev{procedure} terminates outputting $\sigma_{d_{3}+1}$, then $v_{1},v_{2}$, and $v_{3}$ prefer $\sigma_{d_{3}+1}$ to~$\pi$. Otherwise, there does not exist a ranking preferred by all voters to~$\pi$.
    \end{myclaim}
    \begin{proof}

The \dmrev{procedure} terminates before reaching $\sigma_{d_{3}+1}$ if and only if $K(\sigma_j,\pi_{v_3})=0$ for some integer~$j < d_{3}+1$, otherwise by Observation~\ref{obschange}, we can find a swap that is good for~$v_{3}$.\\
If the \dmrev{procedure} terminates before reaching $\sigma_{d_{3}+1}$, necessarily $K(\pi_{v_{3}},\sigma)\leq d_{3}$. Since $K(\pi_{v_{3}},\sigma)-K(\pi_{v_{3}},\pi)=d_{3}$, $K(\pi_{v_{3}},\pi)=0$ i.e.\ $\pi_{v_{3}}=\pi$. So clearly there cannot exist a ranking preferred by all voters, including $v_3$, to~$\pi$.
If we successfully obtain $\sigma_{d_{3}+1}$, it will be at most $d_{3}+1$ swaps away from~$\sigma$. So $\sigma_{d_{3}+1}$ is closer to $\pi_{v_{1}}$ than $\pi$ by
   \begin{eqnarray}
   d_{1}' & := & K(\pi,\pi_{v_{1}})-K(\sigma_{d_{3}+1},\pi_{v_{1}}) \nonumber \\ 
          & \geq & K(\pi,\pi_{v_{1}})-K(\sigma,\pi_{v_{1}})-(d_{3}+1) \nonumber\\
          & = & d_{1}-d_{3}-1 \nonumber \\
          & > & 0 \label{eq:7}
   \end{eqnarray}
where inequality in Line~\ref{eq:7} follows from Claim~\ref{claim9}.
   Similarly $$d_{2}':=K(\pi,\pi_{v_{2}})-K(\sigma_{d_{3}+1},\pi_{v_{2}})>0,$$ that is, $\pi_{v_{2}}$ is closer to $\sigma_{d_{3}+1}$ than to $\pi$.
   
   Also, by construction
    \begin{eqnarray*}
    d_{3}' & := & K(\pi,\pi_{v_{3}})-K(\sigma_{d_{3}+1},\pi_{v_{3}})\\
           & = & K(\pi,\pi_{v_{3}})-K(\sigma,\pi_{v_{3}})+K(\sigma,\pi_{v_{3}})-K(\sigma_{d_{3}+1},\pi_{v_{3}})\\
           & = & -d_{3}+d_{3}+1 \\
           & = & 1.
    \end{eqnarray*}
   So $\pi_{v_{3}}$ is closer to $\sigma_{d_{3}+1}$ than to~$\pi$. That is, all of $v_{1},v_{2}$ and $v_{3}$ prefer to $\sigma_{d_{3}+1}$ to~$\pi$.
   \myqed\end{proof}

\dmrev{\textbf{Case 3}}: Suppose that only one voter, without loss of generality $v_{1}$, prefers $\sigma$ to $\pi$. We show that this case cannot occur. There exists a bubble sort swap on the path from $\pi$ to $\sigma$ that is good for both $v_{2}$ and $v_{3}$, else there cannot be a ranking preferred by all voters by Observation~\ref{obstwo}. Since every bubble sort swap is good for at least one of $v_{2}$ and $v_{3}$, without loss of generality, let $v_{2}$ be the voter for whom at least half of the bubble sort swaps are good. This means that $v_{2}$ has more good swaps on the path than bad swaps, i.e.\ $v_{2}$ also prefers $\sigma$ to $\pi$, a contradiction to $v_1$ being the only voter who prefers $\sigma$ to~$\pi$.

\dmrev{\textbf{Case 4}:} Finally, suppose that no voter prefers $\sigma$ to $\pi$, i.e.\ $K(\pi_{v_{i}},\sigma)\geq K(\pi_{v_{i}},\pi)$ for all $1\leq i \leq 3$. Since $\sigma$ is topologically sorted and hence a Kemeny consensus (see Observation~\ref{fact2}), $\sum_{i=1}^{3} K(\pi_{v_{i}},\sigma)\leq \sum_{i=1}^{3} K( \pi_{v_{i}},\pi)$ holds. From these two inequalities follows that $K(\pi_{v_{i}},\sigma)= K(\pi_{v_{i}},\pi)$ for all $1\leq i \leq 3$. That is, $\pi$ is also a Kemeny consensus, hence there does not exist a ranking preferred to $\pi$ by all voters.

Having discussed all four cases, we now can output a ranking preferred by all the voters to a given ranking $\pi$ \dmrev{or report that no such ranking exists} 
in polynomial time.
\qed\end{proof}
\begin{lemma}\label{eqevenodd1}
For $k \geq 3$, if at least one of $(2k-1)$-\textsc{surv}, $(2k-1)$-\textsc{wurv}, and $(2k-2)$-\textsc{wurv} is polynomial-time solvable, then $k$-\textsc{all-closer-ranking} is polynomial-time solvable.
\end{lemma}

\begin{proof}
Assume first that $(2k-2)$-\dmrev{\textsc{wurv}} \dmrev{is polynomial-time solvable. }
Consider an instance of $k$-\textsc{all-closer-ranking} with input ranking $\pi$ and profile $(\pi_{v_{1}},\ldots,\pi_{v_{k}})$ over $C$.
From this instance of $k$-\textsc{all-closer-ranking} we construct the following instance of $(2k-2)$-\dmrev{\textsc{wurv}}. We copy $\pi$ as the given input ranking and create $2k-2$ voters, $k-2$ of them with ranking $\pi$ and the other $k$ voters corresponding to $v_{1},\ldots,v_{k}$.
Since voters with ranking $\pi$ clearly prefer $\pi$ to any other ranking, if there exists a ranking preferred by \dmrev{an absolute} majority (at least $k$) of the $2k-2$ voters to $\pi$, then these $k$ voters must be $v_{1},\ldots,v_{k}$. If a ranking is preferred by \dmrev{an absolute} majority of the $2k-2$ voters to $\pi$, then it is a solution to $k$-\textsc{all-closer-ranking}. Hence there is a ranking $\sigma$ preferred by \dmrev{an absolute} majority of the $2k-2$ voters
if and only if $\sigma$ is a solution to $k$-\textsc{all-closer-ranking}.
For $(2k-1)$-\dmrev{\textsc{wurv}} and $(2k-1)$-\textsc{surv} we simply add another voter with ranking $\pi$, and otherwise keep the proof intact.
\qed\end{proof}

\begin{lemma}\label{eqevenodd2}Let $k\geq 3$ be a constant. 
If $k$-\textsc{all-closer-ranking} is polynomial-time solvable, then $(2k-1)$-\dmrev{\textsc{wurv}} and $(2k-2)$-\dmrev{\textsc{wurv}} are both polynomial-time solvable.
\end{lemma}

\begin{proof}
If $k$-\textsc{all-closer-ranking} has a polynomial-time algorithm $A$, then we can solve $(2k-2)$-\textsc{wurv} by applying $A$ to each of the $\binom{2k-2}{k}$ voter groups of size $k$, which itself is a polynomial-time procedure if $k$ is a constant. If one of the calls to $A$ returns yes, return yes, else return no. It is easy to see that this procedure returns yes if and only if there is some group of $k$ voters that prefers another ranking, i.e.\ if and only if there is \dmrev{an absolute} majority that prefers another ranking. A similar argument can be applied for $(2k-1)$-\dmrev{\textsc{wurv}}.
\qed\end{proof}

An immediate consequence of Lemmas~\ref{eqevenodd1} and~\ref{eqevenodd2} is the following result.
 
\begin{theorem} \label{45}
 \dmrev{All of $4$-\textsc{wurv}, $4$-\textsc{surv}, $5$-\textsc{wurv} and  $5$-\textsc{surv}\ are polynomial-time solvable if and only if any one of them is polynomial-time solvable.}
\end{theorem}

\begin{proof}
With $k=3$ in Lemma~\ref{eqevenodd1}, \dmrev{if} $5$-\dmrev{\textsc{wurv}} \dmrev{ is polynomial-time solvable, then }
$3$-\textsc{all-closer-ranking} \dmrev{is also polynomial-time solvable. }
Due to Lemma~\ref{eqevenodd2}, \dmrev{the polynomial-time solvability of} $3$-\textsc{all-closer-ranking} 
implies \dmrev{the polynomial-time solvability of} $4$-\dmrev{\textsc{wurv}}.
By a similar argument, \dmrev{the polynomial-time solvability of} $4$-\dmrev{\textsc{wurv}} 
implies \dmrev{the polynomial-time solvability of} $5$-\dmrev{\textsc{wurv}} 
By Theorem~\ref{atmost5}, an analogous result holds for $4$-\textsc{surv} and $5$-\textsc{surv}.
\qed\end{proof}

\subsection{$\NP$-hardness for $k=6$}
\label{sec:n6}

We now improve upon the $\NP$-hardness result of~\cite[Theorem 4]{VSW14} on the search version of 7-\textsc{wurv} from seven voters to six voters, and also extend it to strongly popular rankings with \dmrev{six} or \dmrev{seven} voters. 

\begin{theorem}
\label{th:nsurv_npc}
The search versions of $6$-\textsc{wurv}, $6$-\textsc{surv}, 
and $7$-\textsc{surv} are all $\NP$-hard.
\end{theorem}
\begin{proof}
To prove the claim we modify the proof from~\cite[Theorem 4]{VSW14}, which shows that the search version of $7$-\textsc{wurv} is $\NP$-hard. In that proof, \dmrev{four} of the \dmrev{seven} voters have rankings $\pi_{1},\pi_{2},\pi_{3},\pi_{4}$, respectively, and the remaining \dmrev{three} voters have ranking~$L(\sigma)$. The authors (implicitly) prove that it is $\NP$-hard to construct 
a ranking $\zeta$ that all the \dmrev{four} voters with rankings $\pi_{1},\pi_{2},\pi_{3},\pi_{4}$ prefer to~$L(\sigma)$. \dmrev{We use this \dmrev{to} show \dmrev{the} NP-hardness of $6$-\textsc{wurv} and $7$-\textsc{surv}.}

We start with $7$-\textsc{surv}. For any ranking $\zeta \neq L(\sigma)$, 
the three voters with lists $L(\sigma)$ must prefer $L(\sigma)$ to $\zeta$. In order for $\zeta$ to be more popular than $L(\sigma)$ in the simple sense, ranking $\zeta$ must be preferred to $L(\sigma)$ by more than three voters. This happens if and only if all four voters with rankings $\pi_{1},\pi_{2},\pi_{3},\pi_{4}$ prefer $\zeta$ to~$L(\sigma)$. 

For $6$-\dmrev{\textsc{wurv}}, we have two voters with lists $L(\sigma)$ instead of three. The same reduction holds as for $7$-\textsc{surv}, since \dmrev{an absolute} majority of all six voters, that is, the four voters with rankings $\pi_{1},\pi_{2},\pi_{3},\pi_{4}$, must prefer $\zeta$ to~$L(\sigma)$.

\dmrev{To show \dmrev{the} NP-hardness of $6$-\textsc{surv}, we keep the same instance as for $6$-\textsc{surv}.} Now only two voters prefer $L(\sigma)$ to $\zeta$, and thus $\zeta$ is more popular than $L(\sigma)$ in the simple sense if and only if at least three of the remaining four voters prefer $\zeta$ to $L(\sigma)$, and none of these four voters prefer $L(\sigma)$ to~$\zeta$. Even though it is not directly observed by van Zuylen et al., their $\NP$-hardness proof carries over to this case without modification.
\qed\end{proof}

\section{The relationship with the Kemeny consensus}
\label{sec:kemeny}

\dmrev{We next draw attention to  \dmrev{connections} with the complexity of the famous Kemeny consensus problem.}
\dmrev{We show that if either 
of 4-\textsc{wurv}, 5-\textsc{wurv}, 4-\textsc{surv}, and 5-\textsc{surv}
is polynomial-time solvable, 
then one can find a Kemeny consensus for three voters in polynomial time.}
\dmrev{This explains why we only succeeded to prove polynomial-time solvability for special cases of \dmrev{$k$-\textsc{wurv} and $k$-\textsc{surv}} for $k\in \{4,5\}$ \dmrev{in Lemmas \ref{cl:fourvotersa} and \ref{cl:4votersc}}.}


Consider the following problem: we are given a ranking $\pi$ as well as three voters' rankings $\pi_{v_{1}},\pi_{v_{2}}, \pi_{v_{3}}$. Our task is to output a ranking $\sigma$ that has smaller Kemeny rank than~$\pi$, or \dmrev{report} that none exists. \dmrev{In general, with $k$ voters, we} call this search problem $k$-\textsc{smaller-Kemeny-rank}.
\medskip

\begin{tcolorbox}
\textbf{$k$-\fauxsc{smaller-Kemeny-rank}}\\
\textbf{Input}: A set $C$, a profile $P=(\pi_{v_1},\ldots,\pi_{v_k})$ over $C$ and a ranking $\pi$ over $C$.\\
\textbf{Output}: A ranking $\sigma$ with smaller Kemeny rank than $\pi$, that is, $\sum_{i=1}^{k} K(\sigma,\pi_{v_{i}}) < \sum_{i=1}^{k} K(\pi,\pi_{v_{i}})$ or a statement that no such ranking exists.
\end{tcolorbox} 

\begin{theorem}
\label{th:Kr}
A Kemeny consensus for $k$ voters can be computed in polynomial time if and only if $k$-\textsc{smaller-Kemeny-rank} is \dmrev{polynomial-time solvable}.
\end{theorem}

\begin{proof}
Assume that $k$-\textsc{smaller-Kemeny-rank} has a polynomial-time algorithm~$A$. We simply choose an arbitrary ranking $\pi_1$ for the Kemeny consensus problem and apply $A$ to find $\pi_2$ with smaller Kemeny rank than $\pi_1$, and continue this way until we have found a Kemeny consensus. The number of calls to $A$ can be naively bounded by $k\frac{m(m-1)}{2}$, which is the maximum Kemeny rank of a ranking. Similarly, if we can find a Kemeny consensus for $k$ voters in polynomial time, then we can check if it has smaller Kemeny rank than $\pi$ in the input of the $k$-\textsc{smaller-Kemeny-rank} problem. 
\qed\end{proof}

By an argument similar to the one in~\cite[Theorem~5]{VSW14}, we prove the following result on the complexity of 3-\textsc{smaller-Kemeny-rank}.
 
\begin{theorem} 
\label{th:Kc}
\dmrev{If the search version of 3-\textsc{all-closer-ranking} is polynomial-time solvable then 3-\textsc{smaller-Kemeny-rank} is polynomial-time solvable.}
\end{theorem}

\begin{proof}
Given an instance $I$ of $3$-\textsc{smaller-Kemeny-rank} with profile $P=(\pi_{v_1},\pi_{v_2},\pi_{v_3})$ over $C=\{c_1,\ldots, c_m\}$ and a ranking $\pi$ over $C$, we create an instance $\mathcal{I}'$ of $3$-\textsc{all-closer-ranking} as follows. We create a set of $3m$ candidates as $C'=C^1\cup C^2 \cup C^3$, where $C^j=\{c_r^j : 1\leq r\leq m\}$ for each $1\leq j \leq 3$ and $C^1=C$ with $c_r^1=c_r$ for $1\leq r \leq m$. Intuitively, $C'$ consists of three distinguishable copies of~$C$. Given any ranking $\sigma$ of the $m$ candidates in $\mathcal{I}$, let $\sigma^j$ be the ranking obtained from $\sigma$ by replacing each candidate $c_r$ by $c_r^j$, preserving the original order in~$\sigma$. Let $\pi^{i}$ be a preference ranking of $C^{i}$. We denote by $\pi^{1}\pi^{2}\pi^{3}$ the ranking of $C'$, in which each candidate in $C^{i}$ is preferred to each candidate in $C^{j}$ whenever $i<j$, and candidates within a set $C^{i}$ are ranked according to~$\pi^{i}$. Now the profile $P'=(\pi_{v_1'},\pi_{v_2'},\pi_{v_3'})$ in $\mathcal{I}'$ is defined as follows.
\begin{align*}
\pi_{v_1'} &= \pi_{v_1}^1 \pi_{v_2}^2 \pi_{v_3}^3\\
\pi_{v_2'} &= \pi_{v_2}^1 \pi_{v_3}^2 \pi_{v_1}^3\\
\pi_{v_3'} &= \pi_{v_3}^1 \pi_{v_1}^2 \pi_{v_2}^3
\end{align*} 
Finally we create ranking $\pi' = \pi^1 \pi^2 \pi^3$ for the input to 3-\textsc{all-closer-ranking}.

\begin{myclaim}
    Ranking $\sigma$ in $\mathcal{I}$ has a smaller Kemeny rank than $\pi$ if and only \dmrev{if there exists a} ranking $\sigma'$ in \dmrev{$\mathcal{I}'$} preferred by all of $v'_{1},v'_{2}$, and $v'_{3}$ to~$\pi'$.
\end{myclaim}

\begin{proof} 
Suppose first that $\sigma$ has a smaller Kemeny rank than $\pi$ in $\mathcal{I}$, that is, $$\dmrev{\sum_{j=1}^3 K(\pi_{v_j},\sigma)<\sum_{j=1}^3 K(\pi_{v_j},\pi).}$$ Let $\sigma'=\sigma^{1}\sigma^{2}\sigma^{3}$. Note that for each $1\leq i\leq 3$, 
$$K(\pi_{v'_{i}},\sigma')=\dmrev{\sum_{j=1}^{3} K(\pi_{v_j},\sigma)<\sum_{j=1}^{3} K(\pi_{v_j},\pi)}=K(\pi_{v'_{i}},\pi').$$
So indeed each $v'_{i}$ for $1\leq i \leq 3$ prefers $\sigma'$ to $\pi'$.

For the converse direction, we first informally summarise the argument. We will argue that if there is a ranking $\sigma'$ in $\mathcal{I}'$ preferred to $\pi'$ by all voters, then we can extract a ranking $\sigma$ in $\mathcal{I}$ with smaller Kemeny rank than $\pi$ in two steps. \dmrev{First of all we can break up $\sigma'$ into three different rankings}, each defined on a different candidate set. Secondly, one of these rankings translated back to our instance $\mathcal{I}$ will be a ranking with a smaller Kemeny rank than $\pi$, as we will argue using the averaging principle. This argument relies on every $\pi_{v_{i}}$, $1 \leq i \leq 3$, appearing once in each ``column'' of $\mathcal{I'}$, hence justifying the cyclic shift used in~$\mathcal{I'}$.

Suppose that $\sigma'$ is preferred to $\pi'$ by all three voters $v'_{1}, v'_{2},v'_{3}$. By Lemma~\ref{pres} we can assume that $\sigma'$ preserves \dmrev{$(C^1,C^2,C^3)$}. So we can also assume that $\sigma'=\zeta_{1}^{1}\zeta_{2}^{2}\zeta_{3}^{3}$, where $\zeta_{j}^{j}$ is a ranking of the candidates in $C^j$ for $1\leq j \leq 3$, that is, these are three different rankings. We let $\zeta_{j}^{l}$ be the ranking that is obtained from $\zeta_{j}^{j}$ by replacing candidate $c_{r}^{j}$ with candidate $c_{r}^{l}$ for $1\leq j, l\leq 3$ and $1\leq r\leq m$, preserving the original order in $\zeta_{j}^{j}$. Let $\tau_{j}=\zeta_{j}^{1}\zeta_{j}^{2}\zeta_{j}^{3}$, so that intuitively, we copy the same ranking three times, on different candidate sets. We will show that for some $1 \leq j\leq 3$, $\tau_j$ is also preferred to $\pi'$ by $v'_{1}$.

Notice that $\sum_{i=1}^{3} K(\sigma',\pi_{v'_{i}})=\sum_{j=1}^{3} K(\tau_{j},\pi_{v'_{1}})$. Since $K(\sigma',\pi_{v'_{i}})<K(\pi',\pi_{v'_{i}})$ for all $1\leq i\leq 3$ and $K(\pi',\pi_{v'_{1}})=K(\pi',\pi_{v'_{2}})=K(\pi',\pi_{v'_{3}})$, it follows that
\begin{equation} \label{eq1}
\sum_{j=1}^{3} K(\tau_{j},\pi_{v'_{1}})=\sum_{i=1}^{3} K(\sigma',\pi_{v'_{i}})<\sum_{i=1}^{3} K(\pi',\pi_{v'_{i}})=3 K(\pi',\pi_{v'_{1}}).
\end{equation}
So there must exist an index $j$, $1\leq j \leq 3$, such that $K(\tau_{j},\pi_{v'_{1}})<K(\pi',\pi_{v'_{1}})$. 
But then 
$$
\sum_{i=1}^3 K(\zeta_{j}^1,\pi_{v_{i}})=\sum_{i=1}^3 K(\zeta_{j}^i,\pi_{v_{i}}^{i})=K(\tau_{j},\pi_{v'_{1}})<
K(\pi',\pi_{v'_{1}})=\sum_{i=1}^3 K(\pi,\pi_{v_{i}}),
$$
which means that $\zeta_{j}^{1}$ has smaller Kemeny rank than $\pi$, as desired. 
\myqed\end{proof}
This finishes the proof of our theorem.
\qed\end{proof}
We observe that a slight tweak to the above proofs lets us show $\NP$-hardness for four problems related to $3$-\textsc{all-closer-ranking}.
\begin{observation}
\label{obs:kemeny}
If finding a ranking with a smaller Kemeny rank than a given ranking $\pi$ for \dmrev{three} voters is $\NP$-hard, then finding a ranking $\zeta$ that exactly one / at least one / exactly two / at least two of the \dmrev{three} voters prefer $\pi$, while no voter prefers $\pi$ to $\zeta$ is also $\NP$-hard.
\end{observation}\begin{proof}
We only need to argue why the converse direction still holds with the weaker assumption in Theorem~\ref{th:Kc}. Note that in the proof of Theorem~\ref{th:Kc}, Inequality~\ref{eq1} still holds if only one / at least one / exactly two / at least two of the three voters is non-abstaining and prefers $\sigma'$ to~$\pi'$, while the other voters abstain.
\qed\end{proof}

\begin{corollary}\label{cor:kemeny}
\dmrev{If \dmrev{any of} 4-\textsc{wurv}, 4-\textsc{surv}, 5-\textsc{wurv}, or 5-\textsc{surv} \dmrev{are} polynomial-time solvable, then we \dmrev{can} find a Kemeny consensus for three voters in polynomial time.}
\end{corollary}

\begin{proof}
This proof is illustrated in Figure~\ref{fi:comp}. By Lemma~\ref{eqevenodd1}, if the search version of 4-\dmrev{\textsc{wurv}} or 5-\dmrev{\textsc{wurv}} is \dmrev{polynomial-time solvable}, then the search version of 3-\textsc{all-closer-ranking} is also \dmrev{ polynomial-time solvable}. Now, if the latter is true, then by Theorem~\ref{th:Kc}, 3-\textsc{smaller-Kemeny-rank} is also \dmrev{ polynomial-time solvable}.
This would finally imply that finding a Kemeny consensus 
for \dmrev{three} voters is \dmrev{ polynomial-time solvable}, by Theorem~\ref{th:Kr}. An analogous result holds for 4-\textsc{surv} and 5-\textsc{surv} by Theorem~\ref{atmost5}.
\qed\end{proof}

\begin{figure}[th]
\centering
\begin{tikzpicture}[node distance=1.4cm,every node/.style={fill=white}, align=center]
  \node (start)             [activityStarts]              {4-\dmrev{\textsc{wurv}}, 5-\dmrev{\textsc{wurv}}, 4-\textsc{surv}, 5-\textsc{surv}};
  \node (onCreateBlock)     [process, below of=start]          {search version of 3-\textsc{all-closer-ranking}};
  \node (onStartBlock)      [process, below of=onCreateBlock]   {3-\textsc{smaller-Kemeny-rank}};
  \node (onResumeBlock)     [activityStarts, below of=onStartBlock]   {finding a Kemeny consensus for \dmrev{three} voters};
   
  \draw[->]             (start) --  node[right, text width=1.9cm]
                                   {Lemma~\ref{eqevenodd1}} (onCreateBlock);
  \draw[->]     (onCreateBlock) -- node[right,text width=2cm]{Theorem~\ref{th:Kc}} (onStartBlock);
  \draw[->]      (onStartBlock) -- node[right,text width=2cm]{Theorem~\ref{th:Kr}} (onResumeBlock);

  \end{tikzpicture}
\caption{
If any of 4-\dmrev{\textsc{wurv}}, 5-\dmrev{\textsc{wurv}}, 4-\textsc{surv}, and 5-\textsc{surv} is \dmrev{ polynomial-time solvable}, then via the depicted implications, finding a Kemeny consensus for three voters is \dmrev{ polynomial-time solvable}.
}
\label{fi:comp}
\end{figure}

\section{Summary and open questions}
\label{summary}
We studied weakly popular rankings, defined in \cite{VSW14}, and also introduced the notion of strongly popular rankings analogous to the concept of popularity found in the matching literature, which ignores abstaining voters. 
Then we showed that a ranking $\pi$ is weakly popular if and only if it is strongly popular 
assuming that 
the majority graph of the abstaining voters between $\pi$ and any other ranking $\sigma$ is acyclic.
Using this result we also proved that the two notions of popularity are equivalent for profiles with at most five voters. For profiles with six voters, however, we showed that this equivalence does not hold anymore.

We found the smallest constant $c$ for which any $c$-sorted ranking of a profile is weakly popular. For two or three voters, a topologically sorted ranking turned out always to be popular with respect to both of the popularity notions. For four voters this also holds as long as the majority graph of the voters is a tournament, but it does not hold in general. We explained that the problem of deciding whether there exists a ranking $\sigma$ that is preferred to a given ranking $\pi$ by a simple or absolute majority of voters for profiles with four of five voters boils down to the problem of deciding for three voters whether there is a ranking $\sigma$ \dmrev{that} they all prefer to~$\pi$. This problem, \dmrev{as} we showed, is polynomial-time solvable if the majority graph of the three voters is acyclic, but its complexity is open in general. Importantly, if it were polynomial-time solvable, this would imply the polynomial-time solvability of the well-known Kemeny consensus problem for three voters, whose complexity is currently open.

The study of popular rankings can be extended into various directions. We now list some 
open problems that our work poses, starting with a question already raised by van Zuylen et al.~\cite{VSW14}, which we made some progress on.

\begin{enumerate}
    \item Determine the complexity of deciding whether a popular ranking exists for an instance with arbitrary~$n$. Our Lemmas~\ref{cl:twovoters}, \ref{3voters}, and~\ref{cl:fourvotersa} show that for at most \dmrev{three} voters and for \dmrev{four} voters with an acyclic tournament as the majority graph, the existence of weakly/strongly popular rankings can be checked efficiently. Besides this, Theorem~\ref{tightc} gives a sufficient condition for the existence of a weakly popular ranking for instances with arbitrary~$n$. 
    \item Determine the complexity of $3$-\textsc{all-closer-ranking}.
    \item Construct an example showing that for \emph{any} $n>5$, the two notions of popularity are not equivalent. 
    Theorem~\ref{tightc} might prove to be helpful here.
\end{enumerate}

Finally, popular rankings can be defined and studied in instances where ties in the rankings are allowed, or the rankings are not necessarily complete. Also, besides the Kendall distance, other metrics on rankings can also be applied, such as the Spearman distance\dmrev{~\cite{Spe04}}.

\paragraph{Acknowledgement} We thank Markus Brill for fruitful discussions, and the reviewers of 
\dmrev{earlier versions of this paper for their valuable suggestions, which have greatly helped to improve the presentation of this paper.} Sonja Kraiczy was supported by Undergraduate Research Bursary 19-20-66 from the London Mathematical Society, \dmrev{ by the School of Computing Science, University of Glasgow, by EPSRC studentship EP/T517811/1 and by Merton College, Oxford}. \'{A}gnes Cseh was supported by OTKA grant K128611 and the J\'anos Bolyai Research Fellowship. David Manlove was supported by EPSRC grant EP/P028306/1.  For the purpose of open access, the authors have applied a Creative Commons Attribution (CC BY) licence to any Author Accepted Manuscript version arising from this submission.


\bibliographystyle{abbrv}
\bibliography{bibliography}


\appendix
\section{\dmrev{Calculations relating to Figure \ref{fi:ex} and the proof of Theorem~\ref{example}}}
\label{app:preferences}
\noindent
We remind the reader that the voters' rankings are as follows.
\begin{align*} 
\pi_{v_{1}} & = [1,2,3],[6,4,5],[8,9,7] \\ 
\pi_{v_{2}} & =[2,3,1],[4,5,6],[9,7,8] \\
\pi_{v_{3}} & =[3,1,2],[5,6,4],[7,8,9]\\
\pi_{v_{4}} & =[1,2,3],[4,5,6],[7,8,9] = \sigma_2\\
\pi_{v_{5}} & = [1,2,3],[5,4,6],[9,7,8]\\ 
\pi_{v_{6}} & = [1,2,3],[5,6,4],[7,9,8]
\end{align*}
Consider the two rankings of the candidates $\sigma_{1}=[1,2,3],[5,6,4],\allowbreak[9,7,8]$ and $\sigma_{2}=[1,2,3],[4,5,6],[7,8,9]$. Below we discuss the roles of the voters and we justify them with calculations and observations. Trivially, $v_{4}$ prefers $\sigma_{2}$ to~$\sigma_{1}$.

\paragraph{Voters $v_{1},v_{2},v_{3}$ abstain in the vote between \dmrev{$\sigma_{1}$ and~$\sigma_{2}$}} We first discuss the three impartial voters and justify why they indeed are impartial between $\sigma_{1}$ and $\sigma_{2}$. Note that each of $v_{1},v_{2},v_{3}$ agrees \dmrev{with $\sigma_{1}$} on one triple, and \dmrev{agrees with $\sigma_{2}$} on one triple. For the remaining \dmrev{two triples in each case}, each of these three voters agrees with neither the corresponding triples in $\sigma_{1}$, nor the \dmrev{ones} in $\sigma_{2}$, but instead \dmrev{has} distance $2$ to \dmrev{each} of these. For example, voter $v_{1}$ agrees \dmrev{with both $\sigma_{1}$ and $\sigma_{2}$} on the first triple, but agrees \dmrev{with neither of them} on the other two triples, and \dmrev{ $K([6,4,5],[5,6,4])=K([6,4,5],[4,5,6])=2$} and $K([8,9,7],[9,7,8])=K([8,9,7],[7,8,9])=2$. 
\dmrev{Hence $K(\pi_{v_1},\sigma_1)=K(\pi_{v_1},\sigma_2)=4$.} 
Voter $v_{2}$ agrees with $\sigma_{2}$ on the second triple, \dmrev{and agrees} with $\sigma_{1}$ on the third triple, \dmrev{while the distances to those triples she disagrees with are $K([2,3,1],[1,2,3])$}$=K([4,5,6],[5,6,4])=K([9,7,8],[7,8,9])=2$. 
\dmrev{Hence $K(\pi_{v_2},\sigma_1)=K(\pi_{v_2},\sigma_2)=4$.}  This can be checked similarly for voter~$v_{3}$.

\paragraph{\dmrev{Voters \dmrev{$v_5$ and $v_6$ each prefer} $\sigma_1$ to $\sigma_2$}} \dmrev{Each of $v_5$ and $v_6$ agrees with each of $\sigma_1$ and $\sigma_2$ on the first triple. Further, each of $v_5$ and $v_6$ agrees with $\sigma_1$ on one other triple, but is two swaps away from $\sigma_2$ with respect to the same triple.  For the remaining triple, each of $v_5$ and $v_6$ is one swap away from each of $\sigma_1$ and $\sigma_2$.  Hence $K(\pi_{v_5},\sigma_1)=K(\pi_{v_6},\sigma_1)=1 < 3=K(\pi_{v_5},\sigma_2)=K(\pi_{v_6},\sigma_2)$.}

\section{\dmrev{An instance admitting a weakly popular ranking, but no strongly popular ranking}}
\label{app:nosr}
\noindent
\dmrev{Consider the following profile involving 9 candidates and 8 voters:}
\begin{align*} 
\pi_{v_{1}} & = [1,2,3],[6,4,5],[8,9,7] \\ 
\pi_{v_{2}} & = [2,3,1],[4,5,6],[9,7,8] \\
\pi_{v_{3}} & = [3,1,2],[5,6,4],[7,8,9]\\
\pi_{v_{4}} & = [1,2,3],[4,5,6],[8,9,7]\\
\pi_{v_{5}} & = [1,2,3],[5,6,4],[7,8,9]\\
\pi_{v_{6}} & = [2,3,1],[4,5,6],[7,8,9]\\ 
\pi_{v_{7}} & = [2,3,1],[5,6,4],[8,9,7]\\
\pi_{v_{8}} & = [2,3,1],[5,6,4],[8,9,7] = \pi_{v_{7}}
\end{align*}

Our computer simulations (\dmrev{the program code is available from} \url{https://github.com/SonjaKrai/PopularRankingsExampleCheck}) confirmed that the unique weakly popular ranking is \dmrev{$\pi:= [2,3,1],[5,6,4],[8,9,7] (= \pi_{v_{7}}  = \pi_{v_{8}} )$. 
Let $\sigma=[1,2,3],[4,5,6],[7,8,9]$. \dmrev{Note that} $K(\pi_{v_i},\sigma)=K(\pi_{v_i},\pi)=4$ for \dmrev{$1\leq i\leq 3$}, $2=K(\pi_{v_i},\sigma)<K(\pi_{v_i},\pi)=4$ for $4\leq i \leq 6$ and $6=K(\pi_{v_i},\sigma)>K(\pi_{v_i},\pi)=0$ for \dmrev{$7\leq i\leq 8$}. 
 \dmrev{Hence three voters prefer $\sigma$ to $\pi$, two voters prefer $\pi$ to $\sigma$, and four voters abstain.  It follows that $\sigma$ is more popular than $\pi$ in the simple sense.}}

\end{document}